\documentclass[acmsmall]{acmart}

\AtBeginDocument{%
  }

\setcopyright{acmcopyright}
\copyrightyear{2018}
\acmYear{2018}
\acmDOI{XXXXXXX.XXXXXXX}

\usepackage{algorithm}
\usepackage{algorithmicx}
\usepackage{algpseudocode}
\usepackage{multirow}
\usepackage{color}
\begin{document}

\title{Towards Differential Privacy in Sequential Recommendation: A Noisy Graph Neural Network Approach}

\author{Wentao Hu}

\email{stevenhwt@gmail.com}
\author{Hui Fang}
\authornote{Corresponding author.}
% \authornotemark[2]
\email{fang.hui@mail.shufe.edu.cn}
\affiliation{%
  \institution{Research Institute for Interdisciplinary Sciences and Key Laboratory of Interdisciplinary Research of Computation and Economics, Shanghai University of Finance and Economics}
   \streetaddress{100 Wudong Road}
  \city{Shanghai}
  \country{China}
  \postcode{200433}
}

\renewcommand{\shortauthors}{Wentao Hu and Hui Fang}

\begin{abstract}
  With increasing frequency of high-profile privacy breaches in various online platforms, users are becoming more concerned about their privacy. And recommender system is the core component of online platforms for providing personalized service, consequently, its privacy preservation has attracted great attention. As the gold standard of privacy protection, differential privacy has been widely adopted to preserve privacy in recommender systems. However, existing differentially private recommender systems only consider static and independent interactions, so they cannot apply to sequential recommendation where behaviors are dynamic and dependent. Meanwhile, little attention has been paid on the privacy risk of sensitive user features, most of them only protect user feedbacks. In this work, we propose a novel DIfferentially Private Sequential recommendation framework with a noisy Graph Neural Network approach (denoted as DIPSGNN) to address these limitations. To the best of our knowledge, we are the first to achieve differential privacy in sequential recommendation with dependent interactions. Specifically, in DIPSGNN, we first leverage piecewise mechanism to protect sensitive user features. Then, we innovatively add calibrated noise into aggregation step of graph neural network based on aggregation perturbation mechanism. And this noisy graph neural network can protect sequentially dependent interactions and capture user preferences simultaneously. Extensive experiments demonstrate the superiority of our method over state-of-the-art differentially private recommender systems in terms of better balance between privacy and accuracy.
\end{abstract}





\maketitle

\section{Introduction}\label{sec1}
Recent years have witnessed the tremendous development of various online platforms such as Facebook, Amazon and eBay, they are playing an increasingly important role in users' daily life. And recommender system is the core component of online platforms for providing personalized services, it takes advantage of abundant personal information to recommend items or services that match user preference \citep{eskens2020personal}. The direct access to sensitive personal information makes recommender system a common target of privacy attacks and thus aggravates users privacy concern. Besides, the enactment of General Data Protection Regulation (GDPR) \citep{voigt2017eu} raises users' awareness of privacy and makes it more urgent to devise privacy-preserving recommender systems.

\citet{menkov2020recommendations} find that the actions of arXiv users would be potentially “visible” under targeted attack and they propose to change the privacy settings of the recommender algorithm to mitigate such privacy risk. \citet{fox2020protect} points out that mobile health applications pose new risks to privacy as they need a large volume of health data to be collected, stored and analyzed. And \citep{mcsherry2009differentially,kosinski2013private} show that users' sensitive attributes such as racial information, sexual orientation and political inclinations can be precisely predicted from their interaction history in online platforms by attribute inference attack \citep{zhang2021graph,beigi2020privacy}. Even the outputs of recommender systems are likely to reveal users' sensitive attributes and actual interactions to malicious attackers \citep{jeckmans2013privacy,calandrino2011you}.
In a nutshell, despite the ubiquity of recommender system in various online platforms, it is vulnerable to privacy attacks and may cause leakage of sensitive personal information.

Existing recommender systems (RSs) can be mainly classified into two categories: one category is traditional RSs which includes content-based and collaborative filtering RSs; the other category is sequential RSs \citep{fang2020deep,wang2019sequential}. Traditional RSs model users' historical interactions in a static and independent way, so they can only capture the static long-term preference but ignore to consider short-term interest and the sequential dependencies among user interactions. In contrast, sequential RSs treat user interactions as a dynamic sequence and take sequential dependencies into account to capture both long-term and short-term interest \citep{DBLP:conf/sigir/WangZ0ZW022}. Figure \ref{seqrec} is an illustration of sequential RS, where each user's interactions are indexed chronologically to form a historical interaction sequence. And sequential RSs need to predict users' next interacted items based on their historical behavior sequences \citep{chen2018sequential,kang2018self,sun2019bert4rec,ma2020memory}. Due to the capability of capturing users' dynamic and evolving preferences, sequential RSs are quite important and popular in modern online platforms and have attracted much attention in academia. Therefore, we focus on sequential RSs in this paper and derive to build a privacy-preserving sequential RS that can simultaneously resist to privacy attacks and retain considerable recommendation performance.

\begin{figure}[htb!]
    \centering
    \includegraphics[scale=0.26]{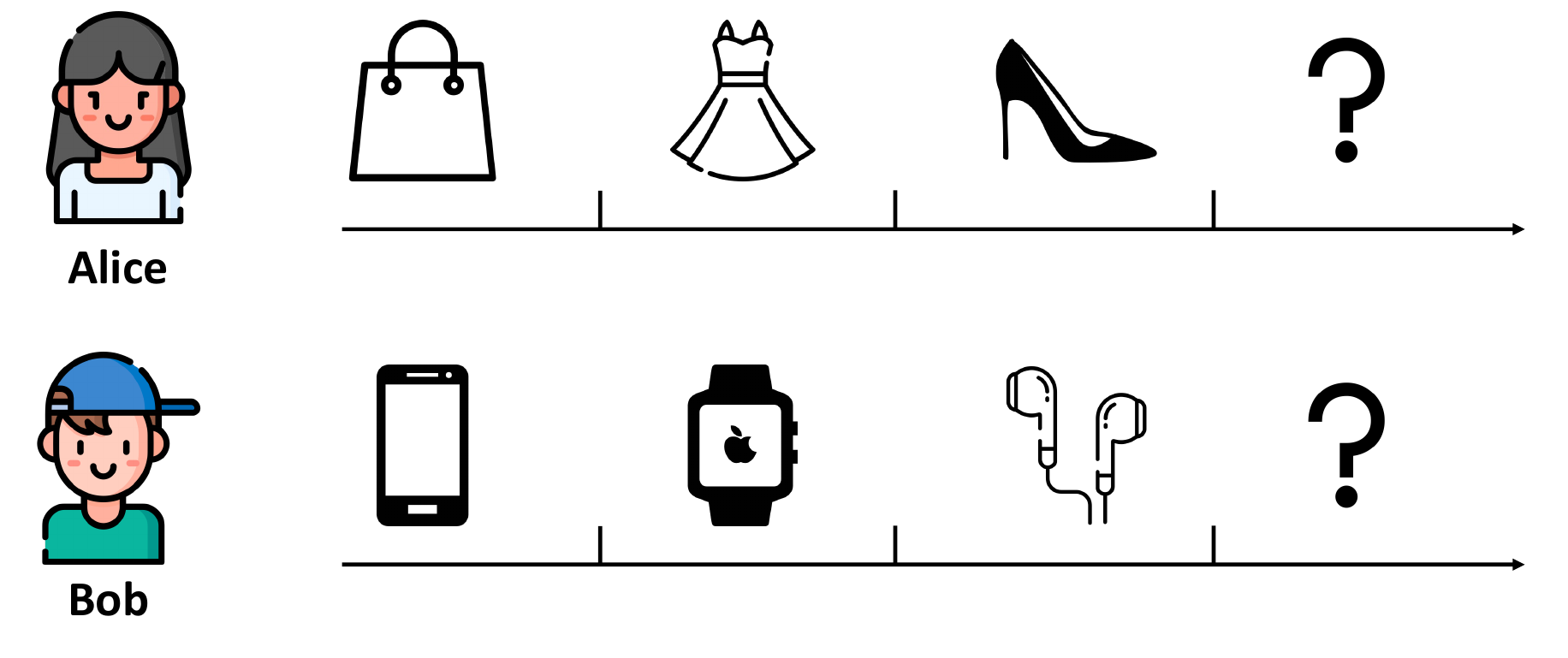}
    \caption{A toy example of sequential RS. Each user' interactions are indexed chronologically to form a interaction sequence. And sequential RSs need to predict the next items that users will interact based on historical interaction sequences.}
    \label{seqrec}
\end{figure}

Previous studies on privacy-preserving RSs mainly adopt anonymisation \citep{duriakova2019pdmfrec}, encryption \citep{erkin2012privacy,nikolaenko2013privacy} and differential privacy \citep{liu2015fast,DBLP:conf/ijcai/HuaXZ15} to protect sensitive user information. The drawback of anonymisation-based RSs is that they cannot provide a provable privacy guarantee or resist to reconstruction attacks \citep{dinur2003revealing,cohen2020linear}. Meanwhile, encryption-based RSs bring heavy computation overhead and fail to prevent attackers from inferring sensitive user information based on the exposed output of RSs \citep{mcsherry2009differentially}. So we resort to differential privacy \citep{dwork2006calibrating} to build a privacy-preserving sequential RS on account of its provable privacy guarantee and lightweight computation overhead. On the other hand, differentially private sequential RSs are quite under-explored because of the challenge to consider sequential dependencies in differential privacy. Existing differentially private RSs are all based on traditional RSs, which can further be divided into two categories. The first category focuses on neighbor-based collaborative filtering \citep{mcsherry2009differentially,zhu2013differential, gao2020dplcf}, where noise is added into the calculation of item/user similarity matrix to protect users' sensitive interactions. The second category is based on matrix factorization \citep{DBLP:conf/ijcai/HuaXZ15,liu2015fast,berlioz2015applying,shin2018privacy}. They add noise into the objective function or gradients in order to protect users' ratings or the fact that a user has rated an item. Despite their effectiveness being partially validated, we argue that these solutions suffer from the following three major limitations.

First, they model users' preferences based on static rating matrix, thus cannot capture dynamic and evolving preferences in users' interaction sequences. Second, interactions are considered to be independent and equally important in existing differentially private RSs. Nevertheless, users' behavior sequences are characterized by complicated sequential dependencies \citep{kang2018self,chang2021sequential} and the most recent interactions may have greater influence on user preferences \citep{fang2020deep}, thus these solutions are not applicable in sequential recommendation. Third, they only protect users' explicit ratings or implicit feedbacks while neglect protection on users' side information such as user demographics. \citet{weinsberg2012blurme} show that there are privacy risks on users' side information, for example, user gender can be accurately inferred from users' ratings. Although \citet{zhang2021graph} design a dual-stage perturbation strategy to protect sensitive user features, they ignore to protect on users' interactions let alone dependent behavior sequences. In short, none of these differentially private RSs focus on protecting users' sensitive features and interactions simultaneously in order to achieve a better balance between privacy and utility.

To bridge the above research gaps, we propose a differentially private sequential recommendation framework called DIPSGNN to protect user features and sequentially dependent interactions at the same time. Specifically, we first take advantage of piecewise mechanism \citep{wang2019collecting} to protect users' sensitive features at input stage and use the perturbed features to initialize user embedding. Then, a gated graph neural network  \citep{DBLP:journals/corr/LiTBZ15} is employed to capture sequential dependencies and dynamic preferences in users' behavior sequences. In this gated graph neural network, we innovatively add calibrated noise into the aggregation step based on aggregation perturbation mechanism \citep{sajadmanesh2022gap} to prevent attackers from inferring users' private interactions based on the exposed recommendation results. To summarize, the main contributions of our work are as follows:
\begin{enumerate}
    \item To the best of knowledge, we are the first to achieve differential privacy for dependent interactions in sequential recommendation.
    
    \item We propose a novel aggregation scheme that can protect time-dependent interactions and capture user preferences without considerably impairing performance.
    
    \item Both users' features and interactions are well-protected in DIPSGNN which offers a better balance between privacy and accuracy.
    
    \item Theoretical analysis and extensive experiments on three real-world datasets demonstrate the effectiveness and superiority of DIPSGNN over state-of-the-art differentially private recommender systems. 
\end{enumerate}

The rest of this article is structured as follows. Section \ref{relatedwork} reviews the related work. Section \ref{preliminaries} introduces the preliminaries and problem setup. Section \ref{methodology} elaborates on the technical details of our proposed DIPSGNN. Section \ref{experiments} discusses the experimental results and analyses. Finally, we conclude this article and propose several future directions in Section \ref{conclusion}.

\section{Related Work}
\label{relatedwork}
In this section, we review three lines of studies related to our work including: sequential recommendation, privacy-preserving recommender systems and privacy-preserving graph neural network.

\subsection{Sequential Recommendation}
Sequential recommendation recommends the next item based on the chronological sequence of users' historical interactions. The earliest work \citep{shani2005mdp} on sequential recommendation leverages Markov decision process to model item transition. Later, FPMC \citep{DBLP:conf/www/RendleFS10} fuses the idea of Markov chains with matrix factorization, it learns the first order transition matrix by assuming next item is only relevant with the previous one. Nevertheless, these conventional methods combine past components independently and neglect long range dependency. For stronger ability to model long-term sequential dependency, deep learning based methods represented by recurrent neural networks (RNN) \citep{quadrana2017personalizing,cui2018mv,xu2019recurrent} and attention mechanism \citep{kang2018self,ying2018sequential,zhang2019feature} have been in blossom in sequential recommendation. For example, \cite{xu2019recurrent} combines the architecture of RNN and Convolutional Neural Network to capture complex dependencies in user behavior sequences. And attention-based neural networks such as Transformer \citep{vaswani2017attention,de2021transformers4rec} and BERT \citep{sun2019bert4rec} use attention scores to 
explore item-item relationships and achieve remarkable performance in sequential recommendation. Recently, graph neural networks (GNN) \citep{DBLP:conf/iclr/KipfW17,velivckovic2018graph,ma2020memory} have attracted much interest in sequential recommendation, as the input data can be represented by graphs. SRGNN \citep{wu2019session} converts user behavior sequences into directed graphs and learns item embeddings on these graphs with gated graph neural network \citep{DBLP:journals/corr/LiTBZ15}. APGNN \citep{zhang2020personalized} promotes SRGNN by fusing personalized user characteristics with item transition patterns in user behavior graphs in order to better model user preferences. SUGRE \citep{chang2021sequential} reconstructs loose item sequences into tight item-item interest graphs based on metric learning, which further improves the performance of GNN in sequential recommendation. By elaborately modeling user interest based on graphs constructed from interaction sequences, GNN-based methods have demonstrated great effectiveness in sequential recommendation.

\subsection{Privacy-preserving recommender systems}

Various privacy-preserving techniques have been explored in recommender systems such as anonymisation \citep{duriakova2019pdmfrec}, encryption \citep{erkin2012privacy,nikolaenko2013privacy}, differential privacy \citep{liu2015fast,DBLP:conf/ijcai/HuaXZ15} and federated learning \citep{li2022federated,li2023distvae,wei2023edge}. However, anonymisation-based RSs cannot provide a provable privacy guarantee and encryption-based RSs bring heavy computation overhead. Besides, they usually focus on traditional recommender systems. Recently, \citep{li2022federated,li2023distvae,wei2023edge} explore to build a federated sequential recommender system in edge computing scenario. But they can only provide general privacy protection by adding noise to the gradient and fail to provide specific privacy guarantee for user features and user behavior sequences in sequential recommendation. As differential privacy can provide provable privacy guarantee and lightweight computation overhead, we focus on utilizing differential privacy to provide exact privacy guarantee for user features and user behavior sequences in sequential recommendation.

Differential privacy \citep{dwork2006calibrating} has been introduced into recommender systems since \citep{mcsherry2009differentially}. It adds noise into calculation of item-similarity matrix in order to protect users' explicit ratings. After that, \citet{gao2020dplcf} protect users' implicit feedbacks by applying binary randomized response \citep{erlingsson2014rappor} on them and then send the perturbed feedbacks to the centralized server to calculate a private item-similarity matrix. 
Aside from neighborhood-based collaborative filtering, there is another line of works which are based on matrix factorization (MF) \citep{koren2009matrix}. \citet{DBLP:conf/ijcai/HuaXZ15} integrate objective perturbation into MF to make sure the final item embedding learned by MF satisfy differential privacy. Besides, they decompose the noise component into small pieces so that it can fit with the decentralized system. \citet{liu2015fast} build a differentially private MF framework by using a novel connection between differential privacy and bayesian posterior sampling via stochastic gradient langevin dynamics. And \citet{shin2018privacy} further divide user ratings into sensitive and non-sensitive ones then add different amount of noise to these two kinds of ratings when calculating gradients. Finally, they achieve a uniform privacy guarantee for sensitive ratings. However, these differentially private recommender systems can only protect static rating matrix and assume that interactions are independent and equally important. They largely ignore the sequential dependencies and users' dynamic preference, which makes them inadequate for sequential recommendation. Meanwhile, the protection on user features is overlooked in these works. Though \cite{zhang2021graph} shed light on the protection towards user demographics, it lacks the ability to protect interactions let alone dependent behavior sequences.

\subsection{Privacy-preserving graph neural network}
Graph neural networks (GNNs) have been broadly employed in sequential recommendation as users' interaction sequences can be easily transformed into graph data. However, rich node features and edge information in GNNs make them vulnerable to privacy attacks. \citep{wu2021linkteller,he2021stealing,zhang2021graphmi} show that private edges can be recovered from GNNs via the influence of particular edge on model output. To mitigate such privacy risk, various privacy-preserving techniques for GNNs are emerging. \citet{wang2021privacy} propose a privacy-preserving representation learning framework on graphs from mutual information perspective. \citet{hsieh2021netfense} perturb graphs based on combinatorial optimization to protect private node labels. Nevertheless, these methods cannot provide a formal privacy guarantee. To address this limitation, differential privacy (DP) has been applied to protect privacy in GNNs. \citet{wu2021linkteller} propose LapGraph to provide DP protection for sensitive edges by adding laplace noise into adjacency matrix. It can be regarded as a data mask method with formal differential privacy guarantee. But \citep{kolluri2022lpgnet,hidano2022degree} argue that adding noise into adjacency matrix destroys original graph structure and ruins the neighborhood aggregation inside GNNs. To remedy this defect, \citet{sajadmanesh2022gap} propose an aggregation perturbation mechanism to safeguard the privacy of edges in undirected and unweighted graphs, which differs from the directed and weighted graphs in our work. More specifically, it forces the sensitivity of embedding update process equal to one by normalizing each row of embedding matrix. However, we find that this normalization step makes embedding matrix deviate too much from its true value and brings excessive noise. To resolve this problem, we normalize the rows of embedding matrix with tunable threshold for different datasets and achieves better utility.

\section{Preliminaries}
\label{preliminaries}
\subsection{Differential Privacy}
Differential privacy \citep{dwork2006calibrating} is a powerful tool to provide formal privacy guarantee when processing sensitive data. It ensures the output of a randomized algorithm is insensitive to the deletion/addition of one individual record in a database by adding calibrated noise to the algorithm. The formal definition of differential privacy is as follows.

\begin{definition}[Differential Privacy]\label{def:dp}
A randomized algorithm $\mathcal{A}$: $\mathcal{X}^{n} \rightarrow \mathcal{Y}$ is $(\epsilon, \delta)$-differentially private, if for all neighboring datasets $X, X^{\prime} \in \mathcal{X}^{n}$ and all $S \subseteq \mathcal{Y}$,
\begin{equation}
\operatorname{Pr}[\mathcal{A}(X) \in S] \leq e^{\epsilon} \cdot \operatorname{Pr}\left[\mathcal{A}\left(X^{\prime}\right) \in S\right]+\delta. 
\end{equation}
\end{definition}
\noindent where $\operatorname{Pr}[\cdot]$ represents probability, $\epsilon>0$ is privacy budget and $\delta>0$ is failure probability. A smaller $\epsilon$ or $\delta$ brings a stronger privacy guarantee but forces us to add more noise in the randomized algorithm. Besides, neighboring datasets denote a pair of datasets differing at most one record. In our work, user interaction sequences are transformed into graphs and we focus on edge-level differential privacy, so neighboring datasets represent two graphs differ at only one edge. Besides graph topology data, our work also involves multidimensional numerical and categorical user feature data.
\citet{ren2022cross} propose a hybrid differential privacy notion to properly perturb heterogeneous data types in social networks. They utilize edge-level differential privacy to protect graph topology data and local differential privacy \cite{kasiviswanathan2011can} to protect user attributes. Inspired by hybrid differential privacy notion, we also integrate local differential privacy into our model to protect user features, as user feature data and graph topology data have different characteristics. The formal definition of local differential privacy is as follows:

\begin{definition}[Local Differential Privacy]\label{def:ldp}
A randomized function $f(\cdot)$ satisfies $\epsilon-\mathrm{LDP}$ if and only if for any two respective inputs $x$ and $x^{\prime}$ of one user and all output $y$, 
\begin{equation}
\operatorname{Pr}\left[f(x)=y\right] \leq e^{\epsilon} \cdot \operatorname{Pr}\left[f\left(x^{\prime}\right)=y\right],
\end{equation}

\end{definition}
\noindent where $\epsilon$ is also called privacy budget. A lower $\epsilon$ provides stronger privacy guarantee but forces us to add heavier noise to each user's data. In local differential privacy, the perturbation of each user's data guarantees that an external attacker cannot easily infer which of any two possible inputs $x$ and $x^{\prime}$ from one user is used to produce the output $y$. Thus, the true input value of this user is protected with high confidence.

\subsection{Problem Statement}
Let $U=\left\{u_{i}\right\}_{i=1}^{|U|}$ and $V=\left\{v_{j}\right\}_{j=1}^{|V|}$ be the set of users and items in the system, respectively. Each user $u$ has a behavior sequence in the chronological order $S^{u}=\left\{v^u_{s}\right\}_{s=1}^{n_{u}}$ ($v^u_s\in V$, and $n_u$ is the length of user $u$'s behavior sequence) and a sensitive feature vector $\mathbf{x}_u$. We convert each $S^u$ into a directed weighted graph $\mathcal{G}^u=(\mathcal{V}^u,\mathcal{E}^u)$, where $\mathcal{V}^u$ and $\mathcal{E}^u$ represent the set of item nodes, and the set of edges, respectively. All numerical features in $\mathbf{x}_u$ are normalized into $[-1,1]$ with min-max normalization and all categorical features are encoded into one-hot vectors. Based on $\{\mathcal{G}^u|u \in U\} $ and $\{\mathbf{x}_u|u \in U\}$, the goal of our work is to \emph{build a differentially private sequential recommendation framework to generate accurate top-$K$ recommendation list for each user ’s next click (i.e., next-item prediction), meanwhile prevent outside malicious attackers from inferring users' sensitive features and sequentially dependent interactions}. 

% It should be noted that the recommender system is \emph{trusted} in our setting.

\section{Methodology}
\label{methodology}
Our proposed DIPSGNN seeks to protect sensitive user features and interactions without sacrificing considerable performance in sequential recommendation. Figure \ref{DIPSGNN_framework} depicts the overview of DIPSGNN: we first protect users' features by perturbing them at input stage. Then, we convert each user's behavior sequence $S^u$ into a user behavior graph $\mathcal{G}^u$ and feed it into DIPSGNN to update user embedding and item embedding. And we add calibrated noise in DIPSGNN layer to prevent the leakage of user interactions from recommendation results. Finally, the updated user embedding and item embedding are concatenated to make next-item recommendation. We will elaborate on details of these components subsequently.

\begin{figure*}[htbp!]
    \centering
    \includegraphics[scale=0.48]{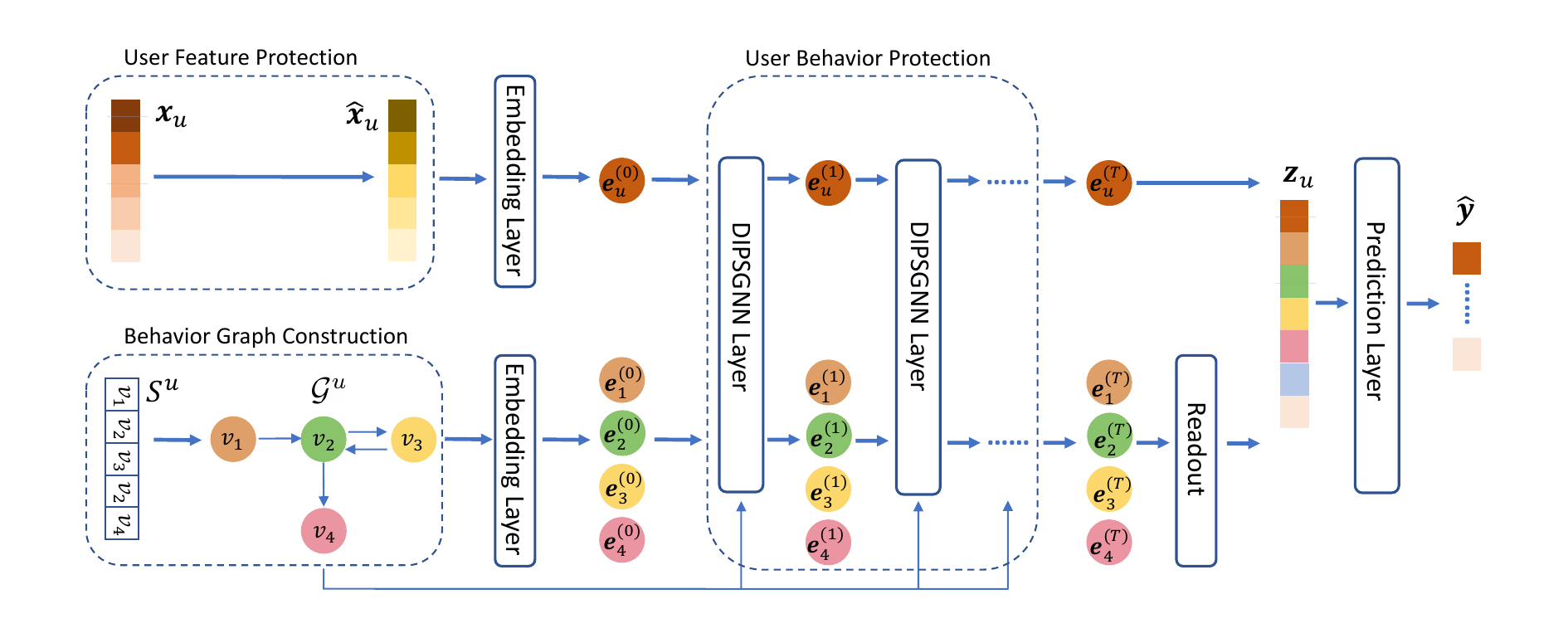}
    \caption{The framework of DIPSGNN. First, user features are perturbed and protected at input stage. Next, we construct user behavior graph based on user interaction sequence. Then, user behavior graph is protected with our newly designed DIPSGNN at embedding update stage. Finally, we utilize updated user embedding and item embedding to make next-item recommendation without leakage of user features and interactions.}
    \label{DIPSGNN_framework}
\end{figure*}

\subsection{User Feature Protection}
To protect sensitive user features, we adopt the strategy to perturb them at input stage with local differential privacy. Concretely, we add noise to raw user features based on piecewise mechanism (PM)  \citep{wang2019collecting}, as it can handle multi-dimensional numerical and categorical features. With the post-processing property of differential privacy, user features will keep private during recommendation. Suppose user $u$'s feature vector consists of $n$ different features, where each numerical feature is represented by a single number and each categorical feature is represented by a single one-hot vector. Thus, user $u$'s feature vector $\mathbf{x_u}=\mathbf{x}_1 \oplus \mathbf{x}_2 \oplus \cdots \oplus \mathbf{x}_n \in \mathbb{R}^{d_0}$ ($d_0\geq n$), where $\oplus$ is the concatenation operation. In this part, we aim to perturb user features with privacy budget $\epsilon_1$. If we perturb each feature equally, then the privacy budget for each feature shrinks to $\epsilon_1/n$. This will harm the utility of the perturbed data as the incurred noise variance is not minimized in this case. To achieve the lowest incurred noise variance, we randomly select $k$ ($k<n$) features from $\mathbf{x}_u$ and perturb each of them with privacy budget $\epsilon_1/k$, while the non-selected features are dropped by masking them with $0$ to avoid privacy leakage. We further follow \citep{wang2019collecting} to set $k$ as:
\begin{equation}
    k=\max \{1, \min \{n,\lfloor\frac{\epsilon_1}{2.5}\rfloor\}\}.
\end{equation}
For each $\mathbf{x}_i$ in the $k$ selected features, if it is a numerical feature, we first normalize it into $[-1,1]$ with min-max normalization and then perturb it by executing Algorithm \ref{pm1} with privacy budget $\epsilon=\frac{\epsilon_1}{k}$. On the contrary, if the selected $\mathbf{x}_i$ is a categorical feature, as it is represented by an one-hot vector, we perturb it with optimized unary encoding (OUE) method \citep{wang2017locally}. Because OUE method will minimize the variance when perturbed one-hot vector has a higher dimension. The details of OUE method are shown in Algorithm \ref{pm2}. By integrating the perturbation for numerical features and categorical features, the whole process of PM are depicted in Algorithm \ref{pm3}. Theorem \ref{pm2_th} guarantees that it satisfies $\epsilon_1$-local differential privacy.

\begin{algorithm}[htb!]
\caption{Piecewise mechanism for one-dimensional number.}
\label{pm1}
\begin{algorithmic}[1] %[1] enables line numbers
\State \textbf{Input}: a number $x \in [-1,1]$ and privacy budget $\epsilon$ 
\State \textbf{Output}: perturbed number ${x'} $
\State  $C \leftarrow \frac{\exp (\epsilon / 2)+1}{\exp (\epsilon / 2)-1}$ 
\State  $l(x)=\frac{C+1}{2} \cdot x-\frac{C-1}{2}$
\State   $r(x)=l(x)+C-1$
        
        \State uniformly sample $\beta$ from $[0,1]$
        \If{$\beta < \frac{\exp(\epsilon/2)}{\exp(\epsilon/2)+1}$ }
            \State uniformly sample $x'$ from $[l(x),r(x)]$
        \Else
            \State uniformly sample ${x'}$ from $[-C,l(x))\cup (r(x),C]$
        \EndIf 
\State \textbf{return} ${x'}$
\end{algorithmic}
\end{algorithm}

\begin{algorithm}[htb!]
\caption{Optimized unary encoding method for one-hot vector.}
\label{pm2}
\begin{algorithmic}[1] %[1] enables line numbers
\State \textbf{Input}: an one-hot vector $\bold{x}=[0,\ldots,0,1,0,\ldots,0] \in \mathbb{R}^{m}$ and privacy budget $\epsilon$ 
\State \textbf{Output}: perturbed one-hot vector ${\bold{x}'}$
\For{$i \in \{1,2,\ldots,m\}$}
        \State $\bold{x}[i]$ and $\bold{x}'[i]$ denote $i$-th position in an one-hot vector
        \State $
            \operatorname{Pr}[{\bold{x}'}[i]=1]=\begin{cases}
           0.5, \ \ \ \ \ \ \ \ \ \ \ \ \ \ \text{if \ $\bold{x}[i]=1$} \\
           \frac{1}{\exp(\epsilon)+1}, \ \ \ \ \ \text{if \ $\bold{x}[i]=0$}
        \end{cases} $
\EndFor
\State \textbf{return} ${\bold{x}'}$
\end{algorithmic}
\end{algorithm}

\begin{algorithm}[htb!]
\caption{Piecewise mechanism for multi-dimensional data with numerical and categorical features.}
\label{pm3}

\begin{algorithmic}[1] %[1] enables line numbers
\State \textbf{Input}:  $\mathbf{x_u}=\mathbf{x}_1 \oplus \mathbf{x}_2 \oplus \cdots \oplus \mathbf{x}_n \in \mathbb{R}^{d_0}$, privacy budget $\epsilon_1$ 
\State \textbf{Output}: perturbed feature vector $\widehat{\mathbf{x}}_u$
 \State  $\mathbf{x_u} = \mathbf{x}_1 \oplus \mathbf{x}_2 \oplus \cdots \oplus \mathbf{x}_n \leftarrow \{0\}^{d_0}$
 \State $k \leftarrow \max \{1, \min \{n,\lfloor\frac{\epsilon_1}{2.5}\rfloor\}\}$
 \State $\mathcal{A} \leftarrow k$ different values uniformly sampled from $\{1,2,\ldots,n\}$
\For{$i \in \mathcal{A}$}
    \If{$\mathbf{x}_i$ is a numerical feature}
        \State $\widehat{\mathbf{x}}_i \leftarrow$ run Algorithm \ref{pm1} with $x=\mathbf{x}_i$ and $\epsilon=\epsilon_1/k$
        \State $\widehat{\mathbf{x}}_i \leftarrow \frac{n}{k} \widehat{\mathbf{x}}_i$
    \Else 
    \State 
    $\widehat{\mathbf{x}}_i \leftarrow$ run Algorithm \ref{pm2} with $\mathbf{x}=\mathbf{x}_i$ and $\epsilon=\epsilon_1/k$     
    \EndIf   
\EndFor
\State \textbf{return} $\widehat{\mathbf{x}}_u$
\end{algorithmic}
\end{algorithm}

\begin{theorem}
\label{pm2_th}
Algorithm \ref{pm3} satisfies $\epsilon_1$-local differential privacy.
\end{theorem}
\begin{proof}
Please see appendix.
\end{proof}

\subsection{Behavior Graph Construction}
To capture the complex sequential dependencies and transition patterns, we convert each user's behavior sequence $S^u$ into a user behavior graph $\mathcal{G}^u=(\mathcal{V}^u,\mathcal{E}^u)$. Inspired by \citep{zhang2020personalized,wu2019session}, $\mathcal{G}^u$ is a directed and weighted graph whose topological structure can be represented by two adjacency matrices, $\mathbf{A}_u^{out}$ and $\mathbf{A}_u^{in}$. The weights in $\mathbf{A}_u^{out}, \mathbf{A}_u^{in}$ are the occurrence number of consecutive interactions between two items. For instance, the weight in position $[i,j]$ of $\mathbf{A}_u^{out}$ is $\text{Count}(v_i,v_j)$, which means the number of times that user $u$ interacts with $v_i$ first, and then immediately with $v_j$. It should be noted that we drop the normalization step \citep{zhang2020personalized,wu2019session} to divide $\text{Count}(v_i,v_j)$ by the outdegree of $v_i$, otherwise the deletion of one interaction in $S^u$ will affect one row rather than one element in $\mathbf{A}_u^{out}$ or $\mathbf{A}_u^{in}$, which impedes the subsequent differential privacy analysis.
\begin{equation}
    \begin{aligned}
        \mathbf{A}_u^{out}[i,j]&=\text{Count}(v_i,v_j), \\
        \mathbf{A}_u^{in}[i,j]&=\text{Count}(v_j,v_i).
    \end{aligned}
\end{equation}

\subsection{User Behavior Protection: DIPSGNN}
As we mentioned before that malicious outside attackers can infer user interactions from recommendation results, we need to add noise into the recommendation algorithm in order to protect interactions. Instead of perturbing user behavior graph at input stage, we choose to add calibrated noise into GNN propagation step to protect user interactions. The reason lies in that perturbation on user behavior graph will destroy original graph structure and distort aggregation process inside GNNs \citep{kolluri2022lpgnet,hidano2022degree}, we will show the superiority of aggregation perturbation over graph structure perturbation by empirical experiments.  
In this section, we will dive into the details of aggregation perturbation inside DIPSGNN. As user characteristics impact their preferences, we consider user features when initializing user embedding. User $u$'s embedding $\mathbf{e}_u^{(0)}$ and item $v_i$'s embedding $\mathbf{e}_i^{(0)}$ are initialized as:
\begin{equation}
    \mathbf{e}_u^{(0)}=\widehat{\mathbf{x}}_u\mathbf{E}_{U},\  \mathbf{e}_i^{(0)}= \mathbf{x}_i\mathbf{E}_{V},
\end{equation}
where $\widehat{\mathbf{x}}_{u} \in \mathbb{R}^{1 \times d_{0}}$ is perturbed feature of user $u$ in Algorithm \ref{pm2} and $\mathbf{E}_{U} \in \mathbb{R}^{d_{0} \times d'}$ is user embedding matrix. Similarly, $\mathbf{x}_{i} \in \mathbb{R}^{1 \times |V|}$ and $\mathbf{E}_{V} \in \mathbb{R}^{|V| \times d}$ are respectively item one-hot encoding and item embedding matrix. Then, we feed the initialized user embedding and item embedding into DIPSGNN and update them iteratively.

At each time step $t$ of node update, we fuse user embedding $\mathbf{e}_u^{(t-1)}$ with item embedding $\mathbf{e}_i^{(t-1)}$ to update them together. $\mathbf{h}^{(t-1)}_i=\mathbf{e}^{(t-1)}_i \oplus \mathbf{e}_u^{(t-1)} \in \mathbb{R}^{1\times (d+d')}$ denotes the joint embedding at $t-1$ time for item $v_i$, where $\oplus$ is the concatenation operation. Combining the joint embedding of all items, we can get a joint embedding matrix $\mathbf{H}^{(t-1)} \in \mathbb{R}^{|V| \times (d+d')}$. 
% Following \cite{sajadmanesh2022gap}, we first get a row-normalized version $\bar{\mathbf{H}}^{(t-1)}$ of $\mathbf{H}^{(t-1)}$,
To bound the sensitivity of joint embedding matrix and facilitate the privacy analysis, we clip each row of $\bar{\mathbf{H}}^{(t-1)}$ to make its norm equal to a constant $C$. And $C$ can be properly tuned on different datasets for better utility.

\begin{equation}
\label{norm}
    \bar{\mathbf{H}}^{(t-1)}_i=\mathbf{h}^{(t-1)}_i * \frac{{C}}{ ||\mathbf{h}^{(t-1)}_i||_2}, \ i=1,\cdots,|V|,
\end{equation}
where $\bar{\mathbf{H}}^{(t-1)}_i$ means $i$-th row of $\bar{\mathbf{H}}^{(t-1)}$. Then, we use sum aggregation to aggregate information from incoming and outgoing neighbors. This step directly accesses the adjacency matrices, $\mathbf{A}_u^{in}$ and $\mathbf{A}_u^{out}$, which contain sensitive structure information of interactions in user behavior sequences. Therefore, we need to add calibrated noise in this step to protect sensitive interactions:
\begin{equation}
    \begin{aligned}
    \widehat{\mathbf{H}}^{(t)}_{out} &= \mathbf{A}_u^{out} \cdot \bar{\mathbf{H}}^{(t-1)}+\mathcal{N}(\sigma^2 \mathbb{I}), \\
    \widehat{\mathbf{H}}^{(t)}_{in} &= \mathbf{A}_u^{in} \cdot \bar{\mathbf{H}}^{(t-1)}+\mathcal{N}(\sigma^2 \mathbb{I}),
    \end{aligned}
\end{equation}
where $\mathcal{N}(\sigma^2 \mathbb{I}) \in \mathbb{R}^{|V| \times (d+d')}$ denotes a noise matrix with each element drawn from Gaussian distribution $\mathcal{N}(0, \sigma^2)$ independently and $\widehat{\mathbf{H}}^{(t)}_{out}, \widehat{\mathbf{H}}^{(t)}_{in}$ are privately aggregated embedding matrices. Theorem \ref{agg_th} will prove that this step satisfies edge-level differential privacy. Besides, the post-processing property of differential privacy \citep{dwork2014algorithmic} guarantees that any operation afterwards will keep private with respect to adjacency matrices, thus we can protect sensitive interactions during recommendation. After neighborhood aggregation, we conduct a linear transformation on the aggregated embedding matrices and get intermediate representation of node $v_i$ as follows: 
\begin{equation}
\begin{aligned}
    \mathbf{a}_{out_{i}}^{(t)} & =(\widehat{\mathbf{H}}^{(t)}_{out} \cdot \mathbf{W}_{out})_i+\mathbf{b}_{out}, \\
    \mathbf{a}_{in_{i}}^{(t)} & =(\widehat{\mathbf{H}}^{(t)}_{in} \cdot \mathbf{W}_{in})_i +\mathbf{b}_{in}, \\
    \mathbf{a}_{i}^{(t)}&=\mathbf{a}_{out_{i}}^{(t)} \oplus \mathbf{a}_{in_{i}}^{(t)},
\end{aligned}
\end{equation}
where $i$ in $(\widehat{\mathbf{H}}^{(t)} \cdot \mathbf{W})_i$ denotes the $i$-th row, $\mathbf{b}_{out}, \mathbf{b}_{in} \in \mathbb{R}^{1\times d}$ are bias terms, and $\mathbf{W}_{out},\mathbf{W}_{in} \in \mathbb{R}^{(d+d')\times d}$ are learnable parameter matrices. $\mathbf{b}_{out}, \mathbf{b}_{in},\mathbf{W}_{out},\mathbf{W}_{in}$ are shared by all users. Then, we leverage gated recurrent unit (GRU) to combine intermediate representation of node $v_i$ with its hidden state of previous time, and then update the hidden state of current time:
\begin{equation}
    \mathbf{e}_{i}^{(t)} =\text{GRU}(\mathbf{a}_{i}^{(t)}, \mathbf{e}_{i}^{(t-1)}).
\end{equation}
It worths noting that $\mathbf{e}_{i}^{(t-1)}$ has been normalized in Equation (\ref{norm}) as a part of $\mathbf{h}_i^{(t-1)}$ and user embedding $\mathbf{e}_u$ will also be updated implicitly in this process. The whole aggregation step in DIPSGNN are shown in Algorithm \ref{agg}, and Theorem \ref{agg_th} proves its privacy guarantee.

\begin{algorithm}[tb]
\caption{Aggregation in DIPSGNN.}
\label{agg}
\begin{algorithmic}[1] %[1] enables line numbers
 \State \textbf{Input}: adjacency matrices $\mathbf{A}_u^{\text{out}}, \mathbf{A}_u^{\text{in}}$; initialized user and item embeddings $\mathbf{e}_u^{(0)}$, $\mathbf{e}_i^{(0)}$; propagation steps $T$; embedding norm $C$; noise standard deviation $\sigma$ 
\State \textbf{Output}: user embedding $\mathbf{e}_u^{(T)}$ and item embedding $\mathbf{e}_{i}^{(T)}$ 

% \State Let $t=0$.
\For{$t=1, \cdots,T$}
        \For{$i=1, \cdots,|V|$} 
        \State $\mathbf{h}^{(t-1)}_i \leftarrow \mathbf{e}^{(t-1)}_i\oplus\mathbf{e}_u^{(t-1)}$
         \State      $\bar{\mathbf{H}}^{(t-1)}_i \leftarrow \mathbf{h}^{(t-1)}_i*\frac{C}{ ||\mathbf{h}^{(t-1)}_i||_2}$   
        \EndFor
        \item[]
        \State $\widehat{\mathbf{H}}^{(t)}_{out} \leftarrow \mathbf{A}_u^{out} \cdot \bar{\mathbf{H}}^{(t-1)}+\mathcal{N}(\sigma^2 \mathbb{I})$  
              
        \State  $\widehat{\mathbf{H}}^{(t)}_{in} \leftarrow \mathbf{A}_u^{in} \cdot \bar{\mathbf{H}}^{(t-1)}+\mathcal{N}(\sigma^2 \mathbb{I})$
        
        \item[]
        \State $\mathbf{a}_{out_{i}}^{(t)} \leftarrow (\widehat{\mathbf{H}}^{(t)}_{out} \cdot \mathbf{W}_{out})_i+\mathbf{b}_{out}$
        
        \State $\mathbf{a}_{in_{i}}^{(t)} \leftarrow (\widehat{\mathbf{H}}^{(t)}_{in} \cdot \mathbf{W}_{in})_i+\mathbf{b}_{in}$
              
        \State $\mathbf{a}_{i}^{(t)} \leftarrow \mathbf{a}_{out_{i}}^{(t)} \oplus \mathbf{a}_{in_{i}}^{(t)}$
        
        \item[]
              
        \State $\mathbf{e}_{i}^{(t)} \leftarrow \text{GRU}(\mathbf{a}_{i}^{(t)}, \mathbf{e}_{i}^{(t-1)})$
        
\EndFor
\State \textbf{return} $\mathbf{e}_{u}^{(T)}$, $\mathbf{e}_{i}^{(T)}$
\end{algorithmic}
\end{algorithm}

\begin{theorem}
\label{agg_th}
For any $\delta \in(0,1)$, propagation steps $T \geq 1$, and noise standard deviation $\sigma>0$, Algorithm \ref{agg} satisfies edge-level $(\epsilon_2, \delta)$-differential privacy with $\epsilon_2=\frac{TC^2}{2 \sigma^{2}}+\frac{C\sqrt{2 T \log (1 / \delta)}}{\sigma}$.
\end{theorem}
\begin{proof}
Please see appendix. 
\end{proof}

\subsection{Prediction and Training}
After finishing the update of all DIPSGNN layers, we get the final representation of all items. Then, we need to obtain a unified representation for each user to conduct next-item prediction. First, we apply a readout function to extract each user's local preference vector $\mathbf{z}_l$ and global preference vector $\mathbf{z}_g$ from item representations.  $\mathbf{z}_l$ is defined as $\mathbf{e}_{n_u}^{(T)}$, which is the final representation of the last item in user $u$'s behavior sequence $S_u$, and $\mathbf{z}_g$ is defined as:
\begin{equation}
    \begin{aligned}
    \alpha_{s} &=\mathbf{q}^{\top} \sigma (\mathbf{W}_{1} \mathbf{e}_{n_u}^{(T)}+\mathbf{W}_{2} \mathbf{e}_{s}^{(T)}+\mathbf{c}), \\
    \mathbf{z}_{\mathrm{g}} &=\sum_{s=1}^{|n_u|} \alpha_{s} \mathbf{e}_{s}^{(T)},
    \end{aligned}
\end{equation}
where $e_s^{(T)}$ refers to the final representation of the $s$-th item in $S_u$, $\mathbf{q} \in \mathbb{R}^d$ and $\mathbf{W}_{1},\mathbf{W}_{2} \in \mathbb{R}^{d \times d} $ are learnable weight matrices. Following \citep{zhang2020personalized}, we concatenate updated user embedding $\mathbf{e}_u^{(T)}$ with the local and global preference vectors, then user $u$'s unified representation $\mathbf{z}_u$ can be expressed as:
\begin{equation}
    \mathbf{z}_u=\mathbf{W}_3 (\mathbf{z}_g \oplus \mathbf{z}_l \oplus \mathbf{e}_u^{(T)}),
\end{equation}
where $\mathbf{W}_{3} \in \mathbb{R}^{d \times (2d+d')}$ is a learnable matrix and $d,d'$ are the dimension of item and user embedding respectively. With user $u$'s unified representation and final representation of all items, we can compute the predicted probability of the next item being $v_i$ by:
\begin{equation}
    \hat{\mathbf{y}}_i=\frac{\exp(\mathbf{z}_u^{\top} \cdot \mathbf{e}_{i}^{(T)})}{\sum_{j=1}^{|V|} \exp(\mathbf{z}_u^{\top} \cdot \mathbf{e}_{j}^{(T)})}.
\end{equation}
And, the loss function is cross-entropy of the prediction $\hat{\mathbf{y}}$ and the ground truth $\mathbf{y}$:
\begin{equation}
    \mathcal{L}(\hat{\mathbf{y}})=-\sum_{i=1}^{|V|} \mathbf{y}_{i} \log \left(\hat{\mathbf{y}}_{i}\right)+\left(1-\mathbf{y}_{i}\right) \log \left(1-\hat{\mathbf{y}}_{i}\right).
\end{equation}
Finally, we use the back-propagation through time (BPTT) algorithm to train the proposed DIPSGNN.

\section{Experiments}
\label{experiments}
\subsection{Experimental Settings}
We evaluate the performance of all methods on three real-world datasets: ML-1M\footnote{\url{grouplens.org/datasets/movielens/}.}, Yelp\footnote{\url{https://www.yelp.com/dataset}.} and Tmall\footnote{\url{tianchi.aliyun.com/dataset/dataDetail?dataId=42}.}. In ML-1M and Tmall, we have 3 categorical user features like age range, gender and occupation, while in Yelp, we have 6 numerical user features like number of ratings, average rating and etc. For Tmall dataset, we use the click data from November 1 to November 7. After obtaining the datasets, we adopt 10-core setting to filter out inactive users and items following \citep{chang2021sequential}. Table \ref{tb_stats} shows their statistics after preprocessing. For each user, we use the first $80\%$ of its behavior sequence as the training set and the remaining $20\%$ constitutes the test set. And hyperparameters are tuned on the validation set which is a random $10\%$ subset of the training set. In training set, we construct subsequences and labels by splitting the behavior sequence of each user. For example, user $u$ has a behavior sequence $S^u=\{v_1^u, v_2^u, \cdots, v_n^u\}$ in training set, we generate a series of subsequences and labels $(\{v_1^u\}, v_2^u), (\{v_1^u,v_2^u\}, v_3^u),\cdots), (\{v_1^u, v_2^u, \cdots, v_{n-1}^u\},v_{n}^u)$, where for example in $(\{v_1^u, v_2^u, \cdots, v_{n-1}^u\},v_{n}^u)$, $\{v_1^u, v_2^u, \cdots, v_{n-1}^u\}$ is the generated subsequence and $v_{n}^u$ is the label of the subsequence, i.e. the next-clicked item. Similarly, we can generate subsequences and labels for validation set and test set.

\begin{table}[htb!]
  \centering
  \caption{Statistics of datasets after preprocessing.}
    \begin{tabular}{cccc}
    \toprule
    Dataset & \multicolumn{1}{c}{ML-1M} & \multicolumn{1}{c}{Yelp} & \multicolumn{1}{c}{Tmall} \\
    \midrule
    Users  & 5,945 & 99,011 & 132,724 \\
    Items  & 2,810 & 56,428 & 98,226 \\
    Interactions  & 365,535 & 1,842,418 & 3,338,788 \\
    Avg. Length  & 96.07 & 27.90  & 36.18 \\
    User features     & 3     & 6     & 2 \\
    \bottomrule
    \end{tabular}%
  \label{tb_stats}%
\end{table}%

To evaluate the performance of DIPSGNN, we compare it with six non-private baselines (BPRMF, SASRec, HRNN, LightGCN, SRGNN, and APGNN) and three private ones (DPMF, DPSGD and EdgeRand):
\begin{itemize}
    \item 
    \textbf{BPRMF} \citep{DBLP:conf/uai/RendleFGS09} is a widely used learning to rank model with a pairwise ranking objective function;
    \item 
    \textbf{SASRec} \citep{kang2018self} is a sequential prediction model based on attention mechanism;
    \item 
    \textbf{HRNN} \citep{quadrana2017personalizing} uses a Hierarchical RNN model to provide personalized prediction in session-based recommendation;
    \item 
    \textbf{LightGCN} \citep{he2020lightgcn} is a simplified graph convolution network for recommendation;
    \item
    \textbf{SRGNN} \citep{wu2019session} utilizes the gated graph neural networks to capture item transitions in user behavior sequences;
    \item
    \textbf{APGNN} \citep{zhang2020personalized} considers user profiles based on SRGNN and captures item transitions in user-specific fashion;
    \item
    \textbf{DPMF}  \citep{DBLP:conf/ijcai/HuaXZ15} is a differentially private matrix factorization method. The original one are based on explicit feedback, we denote it as DPMF\_exp. There is only implicit feedback in Tmall, so we modify it with the same negative sampling strategy as BPRMF to fit in implicit feedback and denote it as DPMF\_imp. We also evaluate DPMF\_imp on other three datasets where ratings larger than $1$ are regarded as positive interactions;
    \item
\textbf{DPSGD} \citep{abadi2016deep} clips the gradient of all parameters in the model and adds noise to the clipped gradient to achieve differential privacy. We incorporate gaussian noise on the clipped gradient to ensure the same level of $(\epsilon,\delta)$-differential privacy as DIPSGNN for a fair comparison. 
    
    \item 
    \textbf{EdgeRand} is a graph structure perturbation method to protect user behaviors and the protection on user features is the same as DIPSGNN. Specifically, we add gaussian noise to the adjacency matrices of user behavior graphs to achieve the same level $(\epsilon,\delta)$-differential privacy as DIPSGNN for fair comparison. It can be seen as a data desensitization method for graphs with formal differential privacy guarantee. The only difference between EdgeRand and existing LapGraph \citep{wu2021linkteller} is that one uses gaussian mechanism and the other uses laplace mechanism. As DIPSGNN uses the notation of $(\epsilon,\delta)$-differential privacy, so we take EdgeRand as our baseline and regard it as the state-of-the-art method to protect edges in GNN following \citep{sajadmanesh2022gap}.
    
\end{itemize}

We aim to recommend top-K items that users are most likely to click on next. If the ground truth next item appears in the recommended top-K items, it means a correct recommendation. Following \citep{wu2019session,zhang2020personalized}, we adopt Recall@$K$ and MRR@$K$ with $K=5,10,20$ as our evaluation metrics. Recall@$K$ represents the proportion of correct recommendations in all samples. MRR@$K$ further considers the rank of the ground-truth next item in the top-K recommendation list, and a larger MRR value indicates that ground-truth next item is at the top of the top-K recommendation list. 

We run all the evaluations $5$ times with different random seeds and report mean value for each method.  The maximum length for user behavior sequences are $100,30,50$ for ML-1M, Yelp and Tmall respectively, which are slightly larger than the average length of user behavior sequences. We set the dimension of item embedding $d=100$ and user embedding dimension $d'=50$ following \citep{zhang2020personalized}, other hyperparameters are tuned to their optimal values on validation set. The model is trained with Adam optimizer and we train DIPSGNN 10 epochs as we observe that our model can reach convergence at that time. For all baseline methods, we use the optimal hyperparameters provided in the original papers.

As for the privacy specification setting, the privacy budget $\epsilon_1$ for protecting user features in EdgeRand and DIPSGNN  is set to $20$ by default following \citep{zhang2021graph}. And the privacy budget $\epsilon_2$ for protecting user behaviors is set to $5$ by default in all private methods, we numerically calibrate the noise standard deviation $\sigma$ according to this privacy budget following \citep{sajadmanesh2022gap}. $\delta$ is set to be smaller than the inverse number of edges. Meanwhile, the different privacy budgets for user features and user behaviors capture the variation of privacy expectations of heterogeneous data types, if all the data types are treated as equally sensitive, it will add too much unneeded noise and sacrifice utility \cite{ren2022cross}. The effect of different privacy budgets will be further discussed by empirical experiments.

\subsection{Performance Comparisons}

Table \ref{tab1} reports the performance comparisons between DIPSGNN and other baselines in terms of Recall@20, MRR@20, Recall@10 and MRR@10.\footnote{We can get similar results regarding $K=5$. Due to space limitation, we do not report.} We have the following observations: (1) the non-private graph based method SRGNN and APGNN achieve the best performance among non-private methods, demonstrating the strong power of adopting graph neural networks to capture complex item transitions and dynamic user preferences in sequential recommendation. (2) As for private baselines, DPMF with explicit feedback outperforms that with implicit feedback, because explicit feedback carries more accurate information about user preferences. But these two DPMF methods both perform much worse than EdgeRand and DIPSGNN. We attribute this phenomenon to that DPMF assumes interactions as independent while EdgeRand and DIPSGNN consider complex dependencies between items by building model on user behavior graphs. This highlights the importance to take into account the dependencies between interactions rather than treating them as independent when protecting user interactions. (3) DPSGD performs much worse than EdgeRand and DIPSGNN, as it adds noise to the gradients of all parameters, distorting the updating process. In contrast, EdgeRand only adds noise to the input, and DIPSGNN only adds noise to the aggregated embedding in GNN. Consequently, they both have a less negative impact on the final performance. (4) Our DIPSGNN consistently yields better performance than the state-of-the-art EdgeRand method for protecting interactions on all three datasets. Their relative gap in terms of Recall@20 is $19.99\%,6.16\%,2.95\%$ on ML-1M, Yelp and Tmall respectively. And their difference is also significant in terms of other metrics, which verifies the effectiveness of DIPSGNN. Moreover, DIPSGNN even outperforms the best non-private baseline APGNN in terms of MRR@20, Recall@10 and MRR@10 on ML-1M. A possible explanation is that controlled amount of noise during training may improve the generalization performance on test set \citep{jim1996analysis}. (5) The performance of DIPSGNN is more competitive than commonly used deep learning recommendation methods such as LightGCN, HRNN and SASRec. And it also beats SRGNN on ML-1M and Yelp. This demonstrates the feasibility of applying DIPSGNN in real-world applications to provide accurate recommendation and protect sensitive user information simultaneously.  

% Table generated by Excel2LaTeX from sheet 'Table1_simplified'
\begin{table*}[htb!]
 % \tiny
  \centering
  \caption{Comparative results of different approaches on all datasets. Bold number means DIPSGNN outperforms EdgeRand. Statistical significance of pairwise differences between DIPSGNN vs. EdgeRand is determined by a paired t-test ($^{*}$ for $p<0.01$). OOM denotes out-of-memory and R denotes Recall, M denotes MRR for succinctness. Nonp means non-private baselines and Priv means private baselines.}
  \resizebox{\textwidth}{27mm}{
    \begin{tabular}{cccccccccccccc}
    \toprule
    \multirow{2}[4]{*}{Type} & \multirow{2}[4]{*}{Method} & \multicolumn{4}{c}{ML-1M}     & \multicolumn{4}{c}{Yelp}      & \multicolumn{4}{c}{Tmall} \\
\cmidrule(lr){3-6} \cmidrule(lr){7-10}  \cmidrule(lr){11-14}        &     & R@20  & M@20  & R@10  & M@10 & R@20  & M@20  & R@10  & M@10 & R@20  & M@20  & R@10  & M@10 \\
    \midrule
    \multirow{6}[2]{*}{Nonp} & BPRMF & 4.56 & 0.81  & 2.20 & 0.65    & 4.38 & 1.01 & 2.58 & 0.89   & 15.19 & 6.77 & 12.28 & 6.57 \\
          & SASRec & 6.69  & 1.20 & 3.44 & 0.98    & 4.70 & 1.03 & 2.69  & 0.89   & 14.46 & 5.08 & 10.46 & 4.81 \\
          & HRNN  & 18.43 & 4.70 & 11.70 & 4.24   & 1.49 & 0.33 & 0.83 & 0.29   & 13.34 & 5.82 & 10.45 & 5.62 \\
          & LightGCN & 8.53 & 1.67 & 4.59 & 1.40   & 6.18 & 1.41 & 3.65 & 1.24 & OOM  & OOM & OOM  & OOM \\
    & SRGNN & 21.01 & 5.15 & 13.09 & 4.61 & 7.24 & 1.77 & 4.41 & 1.58 & 24.85 & 11.55  & 20.42 & 11.24 \\
    & APGNN  & 21.26 & 5.29 & 13.23  & 4.75 & 7.26 & 1.81 & 4.46 & 1.62  & 24.86 & 11.57 & 20.38 & 11.26 \\
    \midrule
    \multirow{3}[2]{*}{Priv} & DPMF\_imp  & 3.15 & 0.66 & 1.77 & 0.57   & 0.41& 0.10 & 0.25 & 0.09   & 1.95 & 0.36 & 1.18  & 0.30 \\
    & DPMF\_exp & 3.40  & 0.71 & 1.89  & 0.61   & 1.30  & 0.32 & 0.83 & 0.28   & -  & - & -  & - \\
    & DPSGD & 4.60  & 0.82 & 2.26  & 0.68   & 0.95  & 0.20 & 0.56 & 0.17   & 1.53  & 0.45 & 0.94  & 0.41 \\
    & EdgeRand & 17.61 & 4.28 &  10.88 & 3.82 & 6.82 & 1.66 & 4.17 & 1.49 & 20.35 & 9.80 & 16.61 & 9.54 \\
    \midrule
    \multirow{2}[2]{*}{Ours} & DIPSGNN & \textbf{21.13} & \textbf{6.11} & \textbf{14.04} & \textbf{5.63} & \textbf{7.24}& \textbf{1.78}  & \textbf{4.45}& \textbf{1.59}  & \textbf{20.95} & \textbf{9.98}  & \textbf{17.15} & \textbf{9.71}  \\
    & Improve & 19.99\%$^{*}$ & 42.76\%$^{*}$ & 29.04\%$^{*}$ & 47.38\%$^{*}$& 6.16\%$^{*}$ & 7.23\%$^{*}$ & 6.71\%$^{*}$& 6.71\%$^{*}$ & 2.95\%$^{*}$ & 1.84\%$^{*}$& 3.25\%$^{*}$ & 1.78\%$^{*}$\\
    \bottomrule
    \end{tabular}%
    }
  \label{tab1}%
\end{table*}%

\subsection{Effect of Privacy Budget}
\label{privbudget}
To analyze the trade-off between privacy and accuracy with different privacy budgets, we first fix the budget for protecting user features $\epsilon_1$ at its default value $20$ and change the budget for protecting user interactions $\epsilon_2$ in $\{3,4,5\}$. The experimental results in terms of Recall@20 are presented in Figure \ref{ep2_recall20}. As Table \ref{tab1} shows, DPMF and DPSGD performs much worse than EdgeRand and DIPSGNN, so we hide them. We can observe that both EdgeRand and DIPSGNN generally perform better with a larger $\epsilon_2$, as the noise added on the adjacency matrices or item embedding matrix in GNN will both decrease when $\epsilon_2$ rises thus brings more accurate recommendation. Meanwhile, DIPSGNN always outperforms EdgeRand except when $\epsilon_2=3$ on Tmall. And the performance gap between them tends to enlarge with a larger $\epsilon_2$,  this confirms again the effectiveness of our proposed DIPSGNN for protecting user interactions. Similarly, we also conduct experiments by changing the budget for user feature protection $\epsilon_1$ in $\{10,20,30\}$ with the fixed privacy budget $\epsilon_2=5$ for user interaction protection. Figure \ref{ep1_recall20} shows experimental results in terms of Recall@20. The performance of EdgeRand and DIPSGNN consistently rises with a larger $\epsilon_1$. As we add less noise to user features when $\epsilon_1$ increases, it will facilitate more accurate modeling of user interests, leading to more satisfying recommendation results. And DIPSGNN outperforms EdgeRand all the time which again shows the superiority.

\begin{figure*}[htb!]
    \centering
    \includegraphics[scale=0.24]{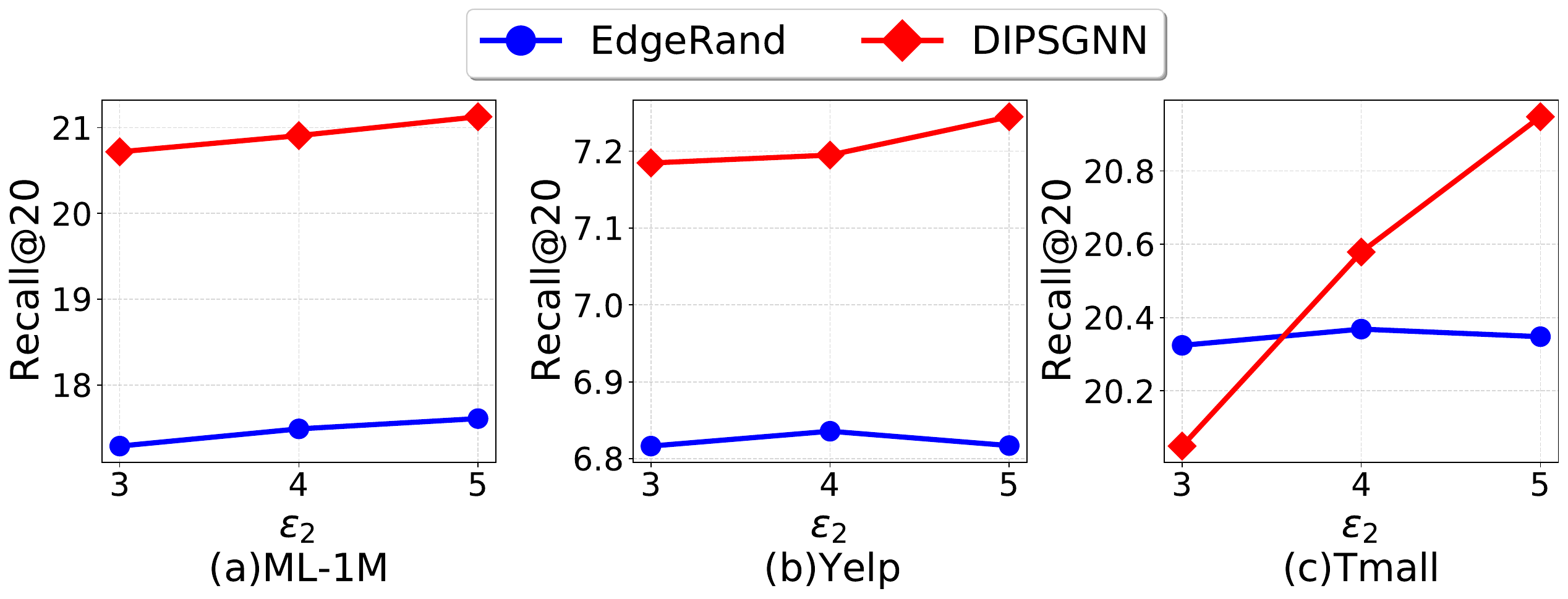}
    \caption{Recall@20 of DIPSGNN and EdgeRand with different $\epsilon_2$ (privacy budget for protect user interactions).}
    \label{ep2_recall20}
\end{figure*}

\begin{figure*}[htbp!]
    \centering
    \includegraphics[scale=0.24]{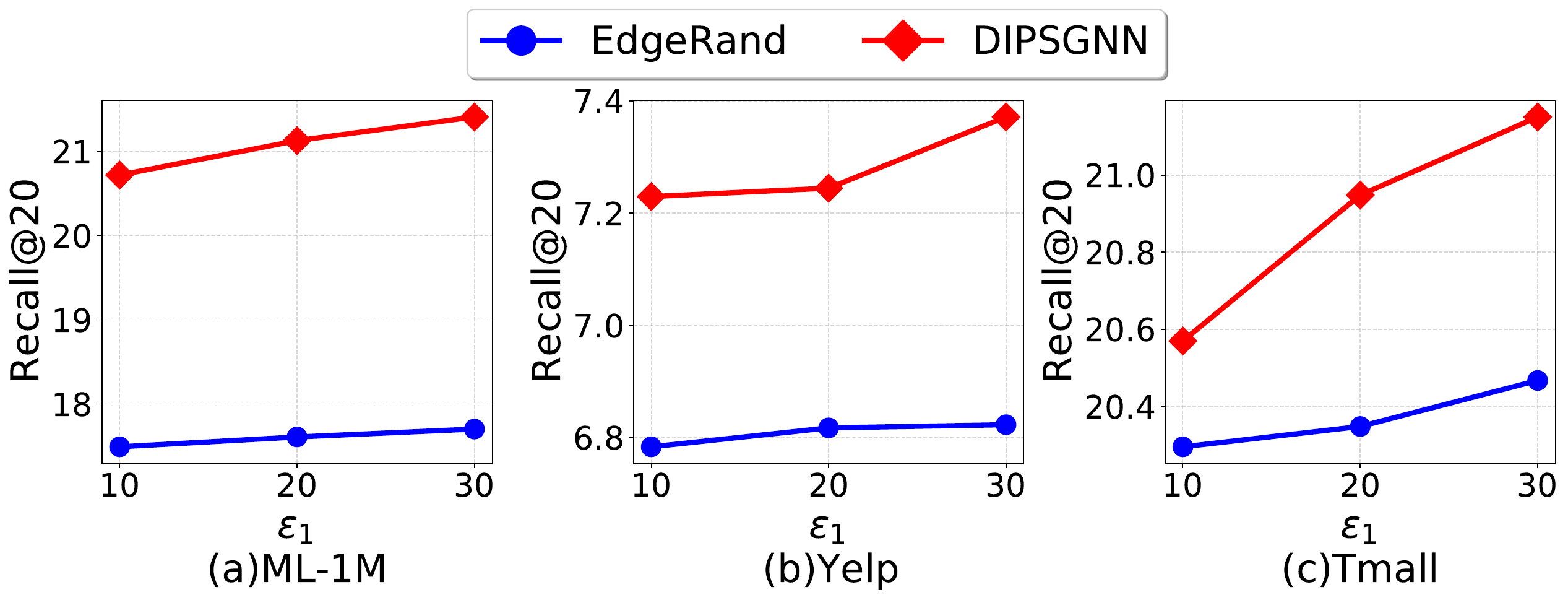}
    \caption{Recall@20 of DIPSGNN and EdgeRand with different $\epsilon_1$ (privacy budget for protect user features).}
    \label{ep1_recall20}
\end{figure*}

\subsection{Hyperparameter Study}

\subsubsection{Effect of GNN Propagation Steps}
In this section, we first study the influence of GNN propagation steps $T$ on the performance of EdgeRand and DIPSGNN. We fix privacy guarantees at their default values and change GNN propagation steps $T$ in $\{1,2,3,4,5\}$. Figure \ref{steps_recall20} shows the experimental results obtained on three datasets. It is evident that the performance of DIPSGNN and EdgeRand both decreases noticeably with larger aggregation steps $T$. Additionally, the performance decline of EdgeRand is apparently more pronounced than that of DIPSGNN, with a dramatic drop observed at steps $T=4$ and $T=5$. This is attributed to EdgeRand adding noise to the original adjacency matrix, distorting neighborhood aggregation inside GNNs \citep{kolluri2022lpgnet,hidano2022degree} and the distortion effect becomes more prominent with more aggregation steps.  On the contrary, DIPSGNN aggregates information from neighbors based on an unperturbed adjacency matrix. The slight performance decrease of DIPSGNN comes from the fact that we need to add more noise to node embedding matrix to maintain the same privacy guarantee, which can be derived from Theorem \ref{agg_th}. 

\begin{figure*}[htb!]
    \centering
    \includegraphics[scale=0.24]{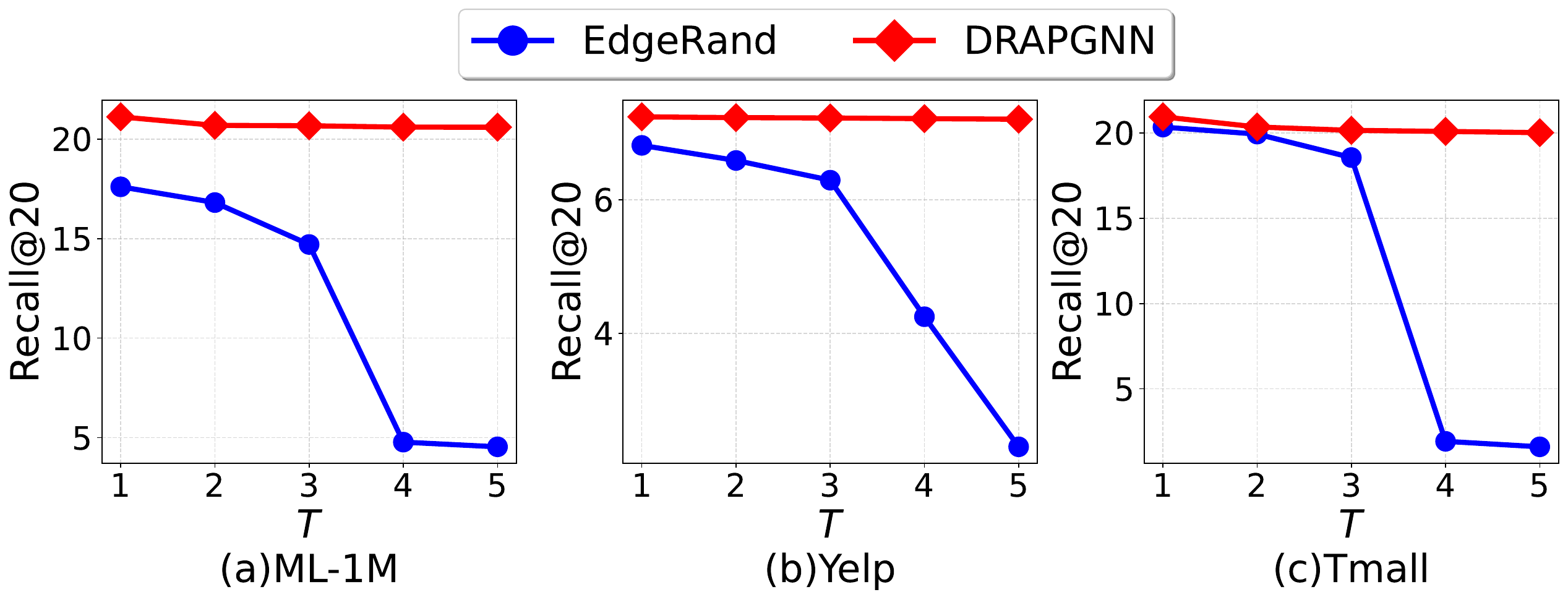}
    \caption{Recall@20 of DIPSGNN and EdgeRand with different GNN aggregation steps $T$.}
    \label{steps_recall20}
\end{figure*}

\subsubsection{Effect of Embedding Norm}
We then investigate how the norm of each row in node embedding matrix affects the performance of DIPSGNN. As we can see from Theorem \ref{agg_th} that a smaller $C$ forces us to add less noise to the embedding matrix at the same privacy guarantee, but if $C$ is too small, the elements of embedding matrix may diverge too much from their true value. To find the proper $C$ for each dataset, we fix privacy guarantees at their default values, GNN propagation steps $T=1$, then select $C$ from $\{0.2,0.4,0.6,0.8,1\}$ for ML-1M,  $\{0.1,0.3,0.5,0.7,0.9\}$ for Yelp and Tmall. The experimental results in terms of Recall@20 are shown in Figure \ref{clip_norm_recall20}. In general, Recall@20 continues to fall with a larger $C$ on ML-1M and Yelp, it reaches the highest value at $C=0.2$ and $C=0.1$ on ML-1M and Yelp respectively. On Tmall, Recall@20 first increases and reaches the highest value at $C=0.3$, then decreases when $C$ becomes larger. We can make a general conclusion that a large $C$ may bring excessive noise and low utility on these three datasets.

\begin{figure*}[htbp!]
    \centering
    \includegraphics[scale=0.24]{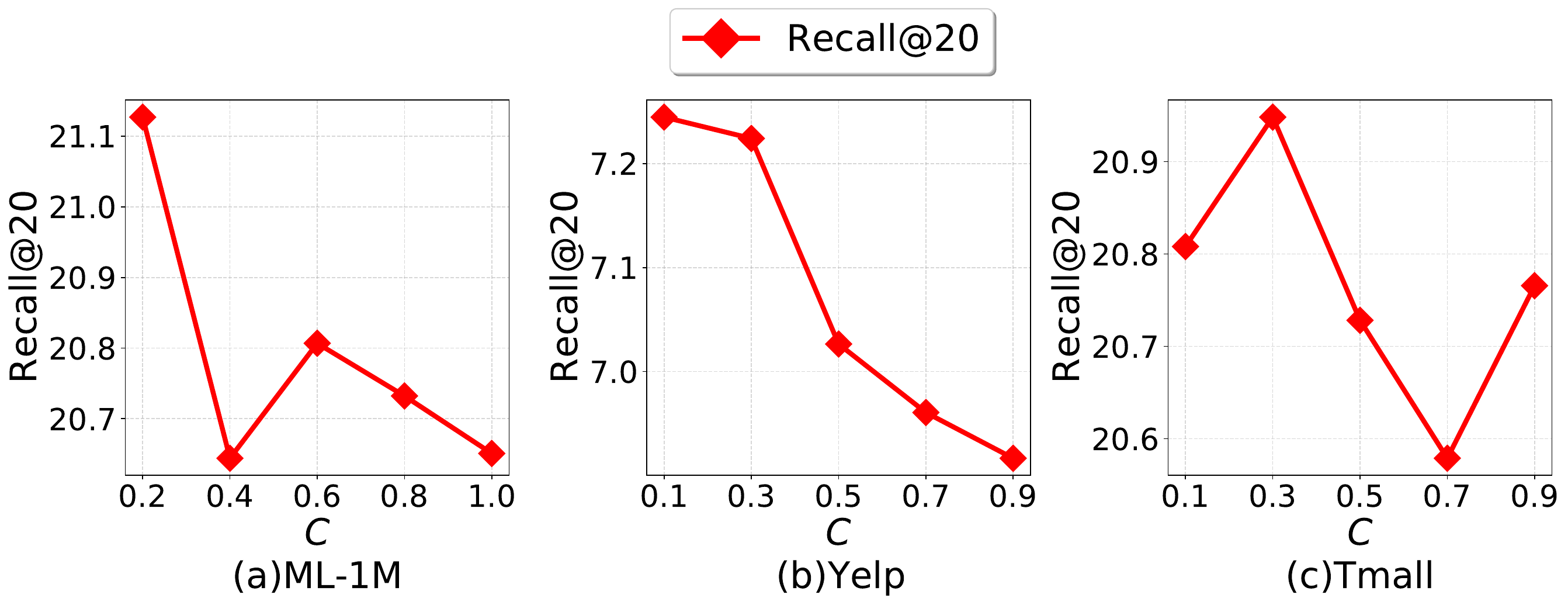}
    \caption{Recall@20 of DIPSGNN with different embedding norm $C$.}
    \label{clip_norm_recall20}
\end{figure*}

\subsection{Importance of User Features}
To verify the necessity of protecting user features to get a better balance between privacy and accuracy, we compare the performance of DIPSGNN with other three variants of DIPSGNN.  \textbf{Nonp} means no noise is added on user features or the training of GNN by setting $\epsilon_1=\infty$ and $\epsilon_2=\infty$. \textbf{Nonp-U} means no user features are exploited and user embedding is initialized by user id. Similarly, \textbf{Priv} means the normal DIPSGNN with perturbed user features and noise added on node embedding matrix during the training of GNN. And in \textbf{Priv-U}, perturbed user features are replaced by user id for initializing user embedding, but the same amount of noise is added on the node embedding matrix during the training of GNN. We show the experimental results in Figure \ref{ablation} and have the following findings. In both non-private and private settings, adding user features can help to capture more accurate user preference and improve recommendation performance. It illustrates the importance of taking user side information, besides user-item interactions, into consideration for a better balance between privacy and accuracy.

\begin{figure*}[htbp!]
    \centering
    \includegraphics[scale=0.28]{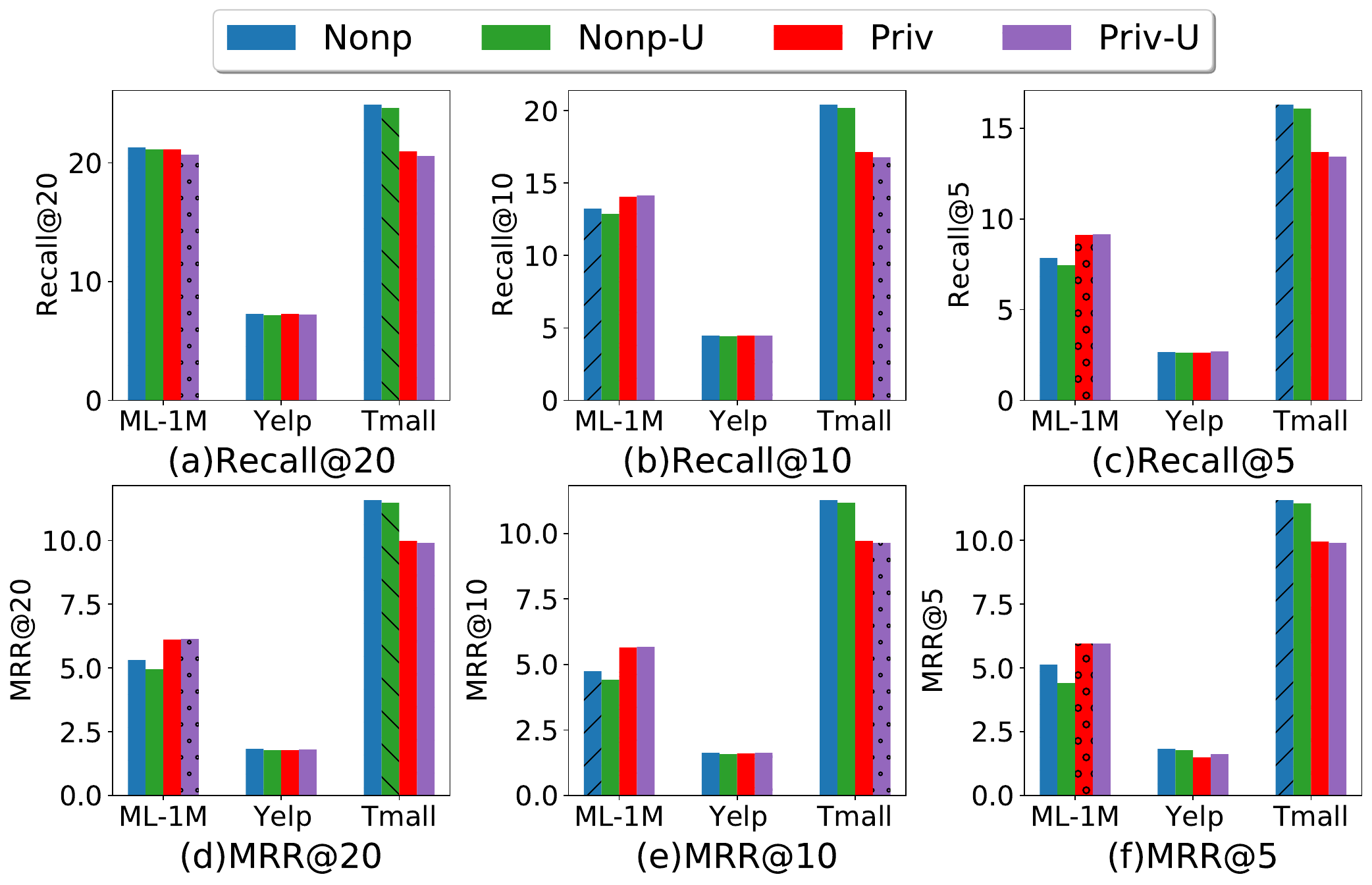}
    \caption{Performance of DIPSGNN compared with three ablation models.}
    \label{ablation}
\end{figure*}

\section{Conclusions and future work}
\label{conclusion}
% With the enactment of General Data Protection Regulation (GDPR), there is an urgent need to protect sensitive user information on various online platforms. And recommender system is the core component of online platforms, it takes advantage of rich personal information to provide personalized service. Therefore, its privacy preservation is of great concern to users and regulators. A privacy-preserving recommender system will greatly alleviate privacy concern and increase user engagement on the platform, thus promote the commercial profit and sustainable development of the platform.

In this paper, we address how to protect sensitive user features and interactions concurrently without great sacrifice of accuracy in sequential recommender system. We propose a differentially private sequential recommendation framework named DIPSGNN. DIPSGNN protects sensitive user features by adding noise to raw features at input stage. The noise scale is determined by piecewise mechanism which can process numerical and categorical features to make them satisfy local differential privacy. And the post-processing property of differential privacy will guarantee that user features are always well protected in the recommendation algorithm. As for the protection on user interactions, we first transform interaction sequences into directed and weighted user behavior graphs. Then, user behavior graphs are fed into gated graph neural network to model sequential dependencies and user preference. In this graph neural network, we design a novel aggregation step to protect adjacency matrices of user behavior graphs, thus to protect user behaviors. Concretely, calibrated noise is added into the aggregation step to make it satisfy differential privacy with respect to adjacency matrices of user behavior graphs. And we empirically demonstrate the superiority of this aggregation perturbation method than conventional graph structure perturbation method for protecting user interactions. Besides, extensive experimental results on three datasets (ML-1M, Yelp and Tmall) evidence that our proposed DIPSGNN can achieve significant gains over state-of-the-art differentially private recommender systems.

For future work, we will extend our framework to other popular graph neural networks such as graph attention networks (GATs) \citep{velivckovic2018graph} and GraphSAGE \citep{hamilton2017inductive}. Besides, we are also interested in incorporating personalized privacy preferences into our framework, as users vary substantially in privacy attitudes in real life \citep{li2020context}.

\section{Acknowledgement}
We greatly acknowledge the support of National Natural Science Foundation of China (Grant No. 72371148 and 72192832), the Shanghai Rising-Star Program (Grant No. 23QA1403100), the Natural Science Foundation of Shanghai (Grant No. 21ZR1421900) and Graduate Innovation Fund of Shanghai University of Finance and Economics (Grant No.CXJJ-2022-366).

\appendix

\section{Proof of Theorem 1}
\begin{proof}
First, we prove Algorithm \ref{pm1} satisfies $\epsilon$-local differential privacy.

In Algorithm \ref{pm1}, $C=\frac{\exp (\epsilon / 2)+1}{\exp (\epsilon / 2)-1}$, $l(x)=\frac{C+1}{2} \cdot x-\frac{C-1}{2}, r(x)=l(x)+C-1$. If $c \in [l(x), r(x)]$, then
\begin{equation*}
\begin{aligned}
    \operatorname{Pr}(x'=c|x)&=\frac{\exp(\epsilon/2)}{\exp(\epsilon/2)+1} \cdot \frac{1}{r(x)-l(x)} \\
    &=\frac{\exp (\epsilon)-\exp (\epsilon / 2)}{2 \exp (\epsilon / 2)+2}=p.
\end{aligned}
\end{equation*}
Similarly, if $c \in [-C,l(x))\cup (r(x),C]$, then
\begin{equation*}
\begin{aligned}
    \operatorname{Pr}(x'=c|x)&=(1-\frac{\exp(\epsilon/2)}{\exp(\epsilon/2)+1}) \cdot \frac{1}{2C-r(x)+l(x)} \\
    &=\frac{\exp (\epsilon/2)-1}{2 \exp (\epsilon)+2 \exp (\epsilon / 2)}=\frac{p}{\exp{(\epsilon)}}.
\end{aligned}
\end{equation*}
Then if $x_1, x_2 \in[-1,1]$ are any two input values and $x' \in[-C, C]$ is the output of Algorithm \ref{pm1}, we have:
\begin{equation*}
    \frac{\operatorname{Pr}(x' \mid x_1)}{\operatorname{Pr}(x' \mid x_2)} \leq \frac{p}{p / \exp (\epsilon)}=\exp(\epsilon).
\end{equation*}
Thus, Algorithm \ref{pm1} satisfies $\epsilon$-LDP. In Algorithm \ref{pm2}, we let $\epsilon=\frac{\epsilon_1}{k}$, so perturbation of numerical feature satisfies $\frac{\epsilon_1}{k}$-LDP.

Analogously, we prove the perturbation of categorical features in Algorithm \ref{pm2} satisfies $\frac{\epsilon_1}{k}$-LDP. Suppose $\mathbf{x}_1$ and $\mathbf{x}_2$ are any two $m$-dimensional one-hot vector for perturbation and the output is $\mathbf{x}'$. In $\mathbf{x}_1$ and  $\mathbf{x}_2$, $\mathbf{x}_1[v_1]=1, \mathbf{x}_2[v_2]=1$  $(v_1 \neq v_2)$ and all other elements are $0$. Let $p=0.5$ and $q=\frac{1}{\exp(\epsilon_1/k)+1}$, $p>q$, we have:
\begin{equation*}
\begin{aligned}
&\frac{\operatorname{Pr}(\mathbf{x}' \mid \mathbf{x}_1)}{\operatorname{Pr}(\mathbf{x}' \mid \mathbf{x}_2)} = \frac{\prod_{i \in[m]} \operatorname{Pr}(\mathbf{x}'[i] \mid \mathbf{x}_1[i])}{\prod_{i \in[m]} \operatorname{Pr}(\mathbf{x}'[i] \mid \mathbf{x}_2[i])}  \\
&=\frac{\operatorname{Pr}(\mathbf{x}'[v_1] \mid \mathbf{x}_1[v_1]) \operatorname{Pr}(\mathbf{x}'[v_2] \mid \mathbf{x}_1[v_2])} {\operatorname{Pr}(\mathbf{x}'[v_1] \mid \mathbf{x}_2[v_1]) \operatorname{Pr}(\mathbf{x}'[v_2] \mid \mathbf{x}_2[v_2])} \\
&\leq \frac{\operatorname{Pr}(\mathbf{x}'[v_1]=1 \mid \mathbf{x}_1[v_1]=1) \operatorname{Pr}(\mathbf{x}'[v_2]=0 \mid \mathbf{x}_1[v_2]=0)}{\operatorname{Pr}(\mathbf{x}'[v_1]=1 \mid \mathbf{x}_2[v_1]=0) \operatorname{Pr}(\mathbf{x}'[v_2]=0 \mid \mathbf{x}_2[v_2]=1)} \\
&=\frac{p}{q} \cdot \frac{1-q}{1-p}=\exp(\frac{\epsilon_1}{k}).
\end{aligned}
\end{equation*}
Thus, the perturbation of categorical feature also satisfies $\frac{\epsilon_1}{k}$-LDP. As Algorithm \ref{pm3} is composed of $k$ times perturbation of numerical or categorical feature, and all of them satisfy $\frac{\epsilon_1}{k}$-LDP, so Algorithm \ref{pm3} satisfies $\epsilon_1$-LDP based on the composition theorem of differential privacy. 
\end{proof}

\section{Proof of Theorem 2}

\begin{proof}

The proof of Theorem \ref{agg_th} uses an alternative definition of differential privacy (DP), called Rényi Differential Privacy (RDP) \citep{mironov2017renyi} which is defined as follows,
\begin{definition}[Rényi Differential Privacy]
A randomized algorithm $\mathcal{A}$ is $(\alpha, \epsilon)$-RDP for $\alpha>1, \epsilon>0$ if for every adjacent datasets $X \sim X^{\prime}$, we have:
$$
D_{\alpha}\left(\mathcal{A}(X) \| \mathcal{A}\left(X^{\prime}\right)\right) \leq \epsilon
$$,
where $D_{\alpha}(P \| Q)$ is the Rényi divergence of order $\alpha$ between probability distributions $P$ and $Q$, defined as:
$$
D_{\alpha}(P \| Q)=\frac{1}{\alpha-1} \log \mathbb{E}_{x \sim Q}\left[\frac{P(x)}{Q(x)}\right]^{\alpha} 
$$
\end{definition}

A basic mechanism to achieve RDP is the Gaussian mechanism. Let $f: \mathcal{X} \rightarrow \mathbb{R}^{d}$ and we first define the sensitivity of $f$ as,
$$
\Delta_{f}=\max _{X \sim X^{\prime}}\left\|f(X)-f\left(X^{\prime}\right)\right\|_{2} 
$$
Then, adding Gaussian noise with variance $\sigma^{2}$ to $f$ as:
$$
\mathcal{A}(X)=f(X)+\mathcal{N}\left(\sigma^{2} \mathbf{e}_{d}\right),
$$
where $\mathcal{N}\left(\sigma^{2} \mathbf{e}_{d}\right) \in  \mathbb{R}^d$ is a vector with each element drawn from $\mathcal{N}(0,\sigma^2)$ independently, yields an $(\alpha, \frac{\Delta_{f}^{2} \alpha}{2 \sigma^{2}})$-RDP algorithm for all $\alpha>1$  \citep{mironov2017renyi}.

As the analysis on $\mathbf{A}_u^{out}$ and $\mathbf{A}_u^{in}$ is the same, so we use $\mathbf{A}$ to denote $\mathbf{A}_u^{out}$ or $\mathbf{A}_u^{in}$. If we delete an interaction from user $u$'s behavior sequence, $\mathbf{A}_u^{out}$ and $\mathbf{A}_u^{in}$ will both only change by 1 at a single position. Let $\mathbf{A}$ and $\mathbf{A}^{'}$ be two neighboring adjacency matrix that only differ by 1 at a single position. Specifically, there exist two nodes $p$ and $q$ such that:
\begin{equation*}
    \begin{cases}
   |\mathbf{A}_{i,j}-\mathbf{A}_{i,j}^{'} |=1 & \text{if}\ i=p\ \text{and}\ j=q, \\
   \mathbf{A}_{i,j}=\mathbf{A}_{i,j}^{'} & \text{otherwise}.
    \end{cases}
\end{equation*}
The sensitivity of sum aggregation step is:
\begin{equation*}
    \begin{aligned}
    &||\mathbf{A} \cdot \bar{\mathbf{H}}^{(t-1)}-\mathbf{A}^{'} \cdot \bar{\mathbf{H}}^{(t-1)}||_F \\
    &=(\sum_{i=1}^{|V|}\|\sum_{j=1}^{|V|}(\mathbf{A}_{i,j} \bar{\mathbf{H}}^{(t-1)}_j-\mathbf{A}_{i, j}^{'}\bar{\mathbf{H}}^{(t-1)}_j)\|_{2}^{2})^{1 / 2} \\
    &=(\|\mathbf{A}_{p, q} \bar{\mathbf{H}}^{(t-1)}_p-\mathbf{A}_{p, q}^{'} \bar{\mathbf{H}}^{(t-1)}_p\|_{2}^{2})^{1 / 2} \\
    &=\|(\mathbf{A}_{p, q}-\mathbf{A}_{p, q}^{'}) \bar{\mathbf{H}}^{(t-1)}_p\|_{2} \\
    &=\|\bar{\mathbf{H}}^{(t-1)}_p\|_{2} \\
    & = C
    \end{aligned}
\end{equation*}
where $\bar{\mathbf{H}}^{(t-1)}_j$ is the $j$-th row of row-normalized feature matrix $\bar{\mathbf{H}}^{(t-1)}$. The sensitivity of each aggregation step is $C$ so it satisfies $(\alpha,C^2\alpha/2\sigma^2)$-RDP based on Gaussian mechanism. And, aggregation in DIPSGNN can be seen as an adaptive composition of $T$ such mechanisms, based on composition property of RDP \citep{mironov2017renyi}, the total privacy cost is $(\alpha,TC^2\alpha/2\sigma^2)$-RDP.

As RDP is a generalization of DP, it can be converted back to standard $(\epsilon,\delta)$-DP using the following lemma.

\begin{lemma}
If $\mathcal{A}$ is an $(\alpha, \epsilon)$-\text{RDP} algorithm, then it also satisfies $(\epsilon+\log (1 / \delta) / \alpha-1, \delta)-D P$ for any $\delta \in(0,1)$.
\end{lemma}

Therefore, $(\alpha,TC^2\alpha/2\sigma^2)$-RDP in DIPSGNN is equivalent to edge-level $(\epsilon_2, \delta)$-DP with $\epsilon_2=\frac{T C^2 \alpha}{2 \sigma^{2}}+\frac{\log (1 / \delta)}{\alpha-1}$. Minimizing this expression over $\alpha>1$ gives $\epsilon_2=\frac{TC^2}{2 \sigma^{2}}+C\sqrt{2 T \log (1 / \delta)} / \sigma$. So we conclude that the aggregation in DIPSGNN satisfies edge-level $(\epsilon_2, \delta)$-DP with $\epsilon_2=\frac{TC^2}{2 \sigma^{2}}+C\sqrt{2 T \log (1 / \delta)} / \sigma$.
\end{proof}

% \appendix


\bibliographystyle{ACM-Reference-Format}
\bibliography{sample-base}

\begin{thebibliography}{70}

%%% the Web address.

\ifx \showCODEN    \undefined \def \showCODEN     #1{\unskip}     \fi
\ifx \showDOI      \undefined \def \showDOI       #1{#1}\fi
\ifx \showISBNx    \undefined \def \showISBNx     #1{\unskip}     \fi
\ifx \showISBNxiii \undefined \def \showISBNxiii  #1{\unskip}     \fi
\ifx \showISSN     \undefined \def \showISSN      #1{\unskip}     \fi
\ifx \showLCCN     \undefined \def \showLCCN      #1{\unskip}     \fi
\ifx \shownote     \undefined \def \shownote      #1{#1}          \fi
\ifx \showarticletitle \undefined \def \showarticletitle #1{#1}   \fi
\ifx \showURL      \undefined \def \showURL       {\relax}        \fi
% invisible to TeX
\providecommand\bibfield[2]{#2}
\providecommand\bibinfo[2]{#2}
\providecommand\natexlab[1]{#1}
\providecommand\showeprint[2][]{arXiv:#2}

\bibitem[\protect\citeauthoryear{Abadi, Chu, Goodfellow, McMahan, Mironov,
  Talwar, and Zhang}{Abadi et~al\mbox{.}}{2016}]%
        {abadi2016deep}
\bibfield{author}{\bibinfo{person}{Martin Abadi}, \bibinfo{person}{Andy Chu},
  \bibinfo{person}{Ian Goodfellow}, \bibinfo{person}{H~Brendan McMahan},
  \bibinfo{person}{Ilya Mironov}, \bibinfo{person}{Kunal Talwar}, {and}
  \bibinfo{person}{Li Zhang}.} \bibinfo{year}{2016}\natexlab{}.
\newblock \showarticletitle{Deep learning with differential privacy}. In
  \bibinfo{booktitle}{\emph{C{C}{S}}}. \bibinfo{pages}{308--318}.
\newblock


\bibitem[\protect\citeauthoryear{Beigi, Mosallanezhad, Guo, Alvari, Nou, and
  Liu}{Beigi et~al\mbox{.}}{2020}]%
        {beigi2020privacy}
\bibfield{author}{\bibinfo{person}{Ghazaleh Beigi}, \bibinfo{person}{Ahmadreza
  Mosallanezhad}, \bibinfo{person}{Ruocheng Guo}, \bibinfo{person}{Hamidreza
  Alvari}, \bibinfo{person}{Alexander Nou}, {and} \bibinfo{person}{Huan Liu}.}
  \bibinfo{year}{2020}\natexlab{}.
\newblock \showarticletitle{Privacy-aware recommendation with private-attribute
  protection using adversarial learning}. In
  \bibinfo{booktitle}{\emph{W{S}{D}{M}}}. \bibinfo{pages}{34--42}.
\newblock


\bibitem[\protect\citeauthoryear{Berlioz, Friedman, Kaafar, Boreli, and
  Berkovsky}{Berlioz et~al\mbox{.}}{2015}]%
        {berlioz2015applying}
\bibfield{author}{\bibinfo{person}{Arnaud Berlioz}, \bibinfo{person}{Arik
  Friedman}, \bibinfo{person}{Mohamed~Ali Kaafar}, \bibinfo{person}{Roksana
  Boreli}, {and} \bibinfo{person}{Shlomo Berkovsky}.}
  \bibinfo{year}{2015}\natexlab{}.
\newblock \showarticletitle{Applying differential privacy to matrix
  factorization}. In \bibinfo{booktitle}{\emph{Rec{S}ys}}.
  \bibinfo{pages}{107--114}.
\newblock


\bibitem[\protect\citeauthoryear{Calandrino, Kilzer, Narayanan, Felten, and
  Shmatikov}{Calandrino et~al\mbox{.}}{2011}]%
        {calandrino2011you}
\bibfield{author}{\bibinfo{person}{Joseph~A Calandrino}, \bibinfo{person}{Ann
  Kilzer}, \bibinfo{person}{Arvind Narayanan}, \bibinfo{person}{Edward~W
  Felten}, {and} \bibinfo{person}{Vitaly Shmatikov}.}
  \bibinfo{year}{2011}\natexlab{}.
\newblock \showarticletitle{" You might also like:" Privacy risks of
  collaborative filtering}. In \bibinfo{booktitle}{\emph{2011 {IEEE} symposium
  on security and privacy}}. IEEE, \bibinfo{pages}{231--246}.
\newblock


\bibitem[\protect\citeauthoryear{Chang, Gao, Zheng, Hui, Niu, Song, Jin, and
  Li}{Chang et~al\mbox{.}}{2021}]%
        {chang2021sequential}
\bibfield{author}{\bibinfo{person}{Jianxin Chang}, \bibinfo{person}{Chen Gao},
  \bibinfo{person}{Yu Zheng}, \bibinfo{person}{Yiqun Hui},
  \bibinfo{person}{Yanan Niu}, \bibinfo{person}{Yang Song},
  \bibinfo{person}{Depeng Jin}, {and} \bibinfo{person}{Yong Li}.}
  \bibinfo{year}{2021}\natexlab{}.
\newblock \showarticletitle{Sequential recommendation with graph neural
  networks}. In \bibinfo{booktitle}{\emph{S{I}{G}{I}{R}}}.
  \bibinfo{pages}{378--387}.
\newblock


\bibitem[\protect\citeauthoryear{Chen, Xu, Zhang, Tang, Cao, Qin, and Zha}{Chen
  et~al\mbox{.}}{2018}]%
        {chen2018sequential}
\bibfield{author}{\bibinfo{person}{Xu Chen}, \bibinfo{person}{Hongteng Xu},
  \bibinfo{person}{Yongfeng Zhang}, \bibinfo{person}{Jiaxi Tang},
  \bibinfo{person}{Yixin Cao}, \bibinfo{person}{Zheng Qin}, {and}
  \bibinfo{person}{Hongyuan Zha}.} \bibinfo{year}{2018}\natexlab{}.
\newblock \showarticletitle{Sequential recommendation with user memory
  networks}. In \bibinfo{booktitle}{\emph{Proceedings of the eleventh ACM
  international conference on {W}eb {S}earch and {D}ata {M}ining}}.
  \bibinfo{pages}{108--116}.
\newblock


\bibitem[\protect\citeauthoryear{Cohen and Nissim}{Cohen and Nissim}{2020}]%
        {cohen2020linear}
\bibfield{author}{\bibinfo{person}{Aloni Cohen} {and} \bibinfo{person}{Kobbi
  Nissim}.} \bibinfo{year}{2020}\natexlab{}.
\newblock \showarticletitle{Linear Program Reconstruction in Practice}.
\newblock \bibinfo{journal}{\emph{Journal of {P}rivacy and {C}onfidentiality}}
  \bibinfo{volume}{10}, \bibinfo{number}{1} (\bibinfo{year}{2020}).
\newblock


\bibitem[\protect\citeauthoryear{Cui, Wu, Liu, Zhong, and Wang}{Cui
  et~al\mbox{.}}{2018}]%
        {cui2018mv}
\bibfield{author}{\bibinfo{person}{Qiang Cui}, \bibinfo{person}{Shu Wu},
  \bibinfo{person}{Qiang Liu}, \bibinfo{person}{Wen Zhong}, {and}
  \bibinfo{person}{Liang Wang}.} \bibinfo{year}{2018}\natexlab{}.
\newblock \showarticletitle{MV-RNN: A multi-view recurrent neural network for
  sequential recommendation}.
\newblock \bibinfo{journal}{\emph{T{K}{D}{E}}} \bibinfo{volume}{32},
  \bibinfo{number}{2} (\bibinfo{year}{2018}), \bibinfo{pages}{317--331}.
\newblock


\bibitem[\protect\citeauthoryear{de~Souza Pereira~Moreira, Rabhi, Lee, Ak, and
  Oldridge}{de~Souza Pereira~Moreira et~al\mbox{.}}{2021}]%
        {de2021transformers4rec}
\bibfield{author}{\bibinfo{person}{Gabriel de Souza Pereira~Moreira},
  \bibinfo{person}{Sara Rabhi}, \bibinfo{person}{Jeong~Min Lee},
  \bibinfo{person}{Ronay Ak}, {and} \bibinfo{person}{Even Oldridge}.}
  \bibinfo{year}{2021}\natexlab{}.
\newblock \showarticletitle{Transformers4Rec: Bridging the Gap between NLP and
  Sequential/Session-Based Recommendation}. In
  \bibinfo{booktitle}{\emph{Rec{S}ys}}. \bibinfo{pages}{143--153}.
\newblock


\bibitem[\protect\citeauthoryear{Dinur and Nissim}{Dinur and Nissim}{2003}]%
        {dinur2003revealing}
\bibfield{author}{\bibinfo{person}{Irit Dinur} {and} \bibinfo{person}{Kobbi
  Nissim}.} \bibinfo{year}{2003}\natexlab{}.
\newblock \showarticletitle{Revealing information while preserving privacy}. In
  \bibinfo{booktitle}{\emph{Proceedings of the twenty-second {ACM
  SIGMOD-SIGACT-SIGART} symposium on {P}rinciples of database systems}}.
  \bibinfo{pages}{202--210}.
\newblock


\bibitem[\protect\citeauthoryear{Duriakova, Tragos, Smyth, Hurley, Pe{\~n}a,
  Symeonidis, Geraci, and Lawlor}{Duriakova et~al\mbox{.}}{2019}]%
        {duriakova2019pdmfrec}
\bibfield{author}{\bibinfo{person}{Erika Duriakova}, \bibinfo{person}{Elias~Z
  Tragos}, \bibinfo{person}{Barry Smyth}, \bibinfo{person}{Neil Hurley},
  \bibinfo{person}{Francisco~J Pe{\~n}a}, \bibinfo{person}{Panagiotis
  Symeonidis}, \bibinfo{person}{James Geraci}, {and} \bibinfo{person}{Aonghus
  Lawlor}.} \bibinfo{year}{2019}\natexlab{}.
\newblock \showarticletitle{PDMFRec: a decentralised matrix factorisation with
  tunable user-centric privacy}. In \bibinfo{booktitle}{\emph{Rec{S}ys}}.
  \bibinfo{pages}{457--461}.
\newblock


\bibitem[\protect\citeauthoryear{Dwork, McSherry, Nissim, and Smith}{Dwork
  et~al\mbox{.}}{2006}]%
        {dwork2006calibrating}
\bibfield{author}{\bibinfo{person}{Cynthia Dwork}, \bibinfo{person}{Frank
  McSherry}, \bibinfo{person}{Kobbi Nissim}, {and} \bibinfo{person}{Adam
  Smith}.} \bibinfo{year}{2006}\natexlab{}.
\newblock \showarticletitle{Calibrating noise to sensitivity in private data
  analysis}. In \bibinfo{booktitle}{\emph{Theory of {C}ryptography
  {C}onference}}. \bibinfo{pages}{265--284}.
\newblock


\bibitem[\protect\citeauthoryear{Dwork, Roth, et~al\mbox{.}}{Dwork
  et~al\mbox{.}}{2014}]%
        {dwork2014algorithmic}
\bibfield{author}{\bibinfo{person}{Cynthia Dwork}, \bibinfo{person}{Aaron
  Roth}, {et~al\mbox{.}}} \bibinfo{year}{2014}\natexlab{}.
\newblock \showarticletitle{The algorithmic foundations of differential
  privacy}.
\newblock \bibinfo{journal}{\emph{Foundations and {T}rends{\textregistered} in
  {T}heoretical {C}omputer {S}cience}} \bibinfo{volume}{9},
  \bibinfo{number}{3--4} (\bibinfo{year}{2014}), \bibinfo{pages}{211--407}.
\newblock


\bibitem[\protect\citeauthoryear{Erkin, Beye, Veugen, and Lagendijk}{Erkin
  et~al\mbox{.}}{2012}]%
        {erkin2012privacy}
\bibfield{author}{\bibinfo{person}{Zekeriya Erkin}, \bibinfo{person}{Michael
  Beye}, \bibinfo{person}{Thijs Veugen}, {and} \bibinfo{person}{Reginald~L
  Lagendijk}.} \bibinfo{year}{2012}\natexlab{}.
\newblock \showarticletitle{Privacy-preserving content-based recommender
  system}.
\newblock In \bibinfo{booktitle}{\emph{Proceedings of the on Multimedia and
  security}}. \bibinfo{pages}{77--84}.
\newblock


\bibitem[\protect\citeauthoryear{Erlingsson, Pihur, and Korolova}{Erlingsson
  et~al\mbox{.}}{2014}]%
        {erlingsson2014rappor}
\bibfield{author}{\bibinfo{person}{{\'U}lfar Erlingsson},
  \bibinfo{person}{Vasyl Pihur}, {and} \bibinfo{person}{Aleksandra Korolova}.}
  \bibinfo{year}{2014}\natexlab{}.
\newblock \showarticletitle{Rappor: Randomized aggregatable privacy-preserving
  ordinal response}. In \bibinfo{booktitle}{\emph{C{C}{S}}}.
  \bibinfo{pages}{1054--1067}.
\newblock


\bibitem[\protect\citeauthoryear{Eskens}{Eskens}{2020}]%
        {eskens2020personal}
\bibfield{author}{\bibinfo{person}{Sarah Eskens}.}
  \bibinfo{year}{2020}\natexlab{}.
\newblock \showarticletitle{The personal information sphere: An integral
  approach to privacy and related information and communication rights}.
\newblock \bibinfo{journal}{\emph{Journal of the {A}ssociation for
  {I}nformation {S}cience and {T}echnology}} \bibinfo{volume}{71},
  \bibinfo{number}{9} (\bibinfo{year}{2020}), \bibinfo{pages}{1116--1128}.
\newblock


\bibitem[\protect\citeauthoryear{Fang, Zhang, Shu, and Guo}{Fang
  et~al\mbox{.}}{2020}]%
        {fang2020deep}
\bibfield{author}{\bibinfo{person}{Hui Fang}, \bibinfo{person}{Danning Zhang},
  \bibinfo{person}{Yiheng Shu}, {and} \bibinfo{person}{Guibing Guo}.}
  \bibinfo{year}{2020}\natexlab{}.
\newblock \showarticletitle{Deep learning for sequential recommendation:
  Algorithms, influential factors, and evaluations}.
\newblock \bibinfo{journal}{\emph{TOIS}} \bibinfo{volume}{39},
  \bibinfo{number}{1} (\bibinfo{year}{2020}), \bibinfo{pages}{1--42}.
\newblock


\bibitem[\protect\citeauthoryear{Fox}{Fox}{2020}]%
        {fox2020protect}
\bibfield{author}{\bibinfo{person}{Grace Fox}.}
  \bibinfo{year}{2020}\natexlab{}.
\newblock \showarticletitle{“To protect my health or to protect my health
  privacy?” A mixed-methods investigation of the privacy paradox}.
\newblock \bibinfo{journal}{\emph{Journal of the Association for Information
  Science and Technology}} \bibinfo{volume}{71}, \bibinfo{number}{9}
  (\bibinfo{year}{2020}), \bibinfo{pages}{1015--1029}.
\newblock


\bibitem[\protect\citeauthoryear{Gao, Huang, Lin, Jin, and Li}{Gao
  et~al\mbox{.}}{2020}]%
        {gao2020dplcf}
\bibfield{author}{\bibinfo{person}{Chen Gao}, \bibinfo{person}{Chao Huang},
  \bibinfo{person}{Dongsheng Lin}, \bibinfo{person}{Depeng Jin}, {and}
  \bibinfo{person}{Yong Li}.} \bibinfo{year}{2020}\natexlab{}.
\newblock \showarticletitle{DPLCF: differentially private local collaborative
  filtering}. In \bibinfo{booktitle}{\emph{S{I}{G}{I}{R}}}.
  \bibinfo{pages}{961--970}.
\newblock


\bibitem[\protect\citeauthoryear{Hamilton, Ying, and Leskovec}{Hamilton
  et~al\mbox{.}}{2017}]%
        {hamilton2017inductive}
\bibfield{author}{\bibinfo{person}{Will Hamilton}, \bibinfo{person}{Zhitao
  Ying}, {and} \bibinfo{person}{Jure Leskovec}.}
  \bibinfo{year}{2017}\natexlab{}.
\newblock \showarticletitle{Inductive representation learning on large graphs}.
\newblock \bibinfo{journal}{\emph{Advances in {N}eural {I}nformation
  {P}rocessing {S}ystems}}  \bibinfo{volume}{30} (\bibinfo{year}{2017}).
\newblock


\bibitem[\protect\citeauthoryear{He, Deng, Wang, Li, Zhang, and Wang}{He
  et~al\mbox{.}}{2020}]%
        {he2020lightgcn}
\bibfield{author}{\bibinfo{person}{Xiangnan He}, \bibinfo{person}{Kuan Deng},
  \bibinfo{person}{Xiang Wang}, \bibinfo{person}{Yan Li},
  \bibinfo{person}{Yongdong Zhang}, {and} \bibinfo{person}{Meng Wang}.}
  \bibinfo{year}{2020}\natexlab{}.
\newblock \showarticletitle{Lightgcn: Simplifying and powering graph
  convolution network for recommendation}. In
  \bibinfo{booktitle}{\emph{S{I}{G}{I}{R}}}. \bibinfo{pages}{639--648}.
\newblock


\bibitem[\protect\citeauthoryear{He, Jia, Backes, Gong, and Zhang}{He
  et~al\mbox{.}}{2021}]%
        {he2021stealing}
\bibfield{author}{\bibinfo{person}{Xinlei He}, \bibinfo{person}{Jinyuan Jia},
  \bibinfo{person}{Michael Backes}, \bibinfo{person}{Neil~Zhenqiang Gong},
  {and} \bibinfo{person}{Yang Zhang}.} \bibinfo{year}{2021}\natexlab{}.
\newblock \showarticletitle{Stealing links from graph neural networks}. In
  \bibinfo{booktitle}{\emph{U{S}{E}{N}{I}{X} {S}ecurity}}.
  \bibinfo{pages}{2669--2686}.
\newblock


\bibitem[\protect\citeauthoryear{Hidano and Murakami}{Hidano and
  Murakami}{2022}]%
        {hidano2022degree}
\bibfield{author}{\bibinfo{person}{Seira Hidano} {and} \bibinfo{person}{Takao
  Murakami}.} \bibinfo{year}{2022}\natexlab{}.
\newblock \showarticletitle{Degree-Preserving Randomized Response for Graph
  Neural Networks under Local Differential Privacy}.
\newblock \bibinfo{journal}{\emph{arXiv preprint arXiv:2202.10209}}
  (\bibinfo{year}{2022}).
\newblock


\bibitem[\protect\citeauthoryear{Hsieh and Li}{Hsieh and Li}{2021}]%
        {hsieh2021netfense}
\bibfield{author}{\bibinfo{person}{I-Chung Hsieh} {and}
  \bibinfo{person}{Cheng-Te Li}.} \bibinfo{year}{2021}\natexlab{}.
\newblock \showarticletitle{NetFense: Adversarial Defenses against Privacy
  Attacks on Neural Networks for Graph Data}.
\newblock \bibinfo{journal}{\emph{T{K}{D}{E}}} (\bibinfo{year}{2021}).
\newblock


\bibitem[\protect\citeauthoryear{Hua, Xia, and Zhong}{Hua
  et~al\mbox{.}}{2015}]%
        {DBLP:conf/ijcai/HuaXZ15}
\bibfield{author}{\bibinfo{person}{Jingyu Hua}, \bibinfo{person}{Chang Xia},
  {and} \bibinfo{person}{Sheng Zhong}.} \bibinfo{year}{2015}\natexlab{}.
\newblock \showarticletitle{Differentially Private Matrix Factorization}. In
  \bibinfo{booktitle}{\emph{I{J}{C}{A}{I}}}.
\newblock


\bibitem[\protect\citeauthoryear{Jeckmans, Beye, Erkin, Hartel, Lagendijk, and
  Tang}{Jeckmans et~al\mbox{.}}{2013}]%
        {jeckmans2013privacy}
\bibfield{author}{\bibinfo{person}{Arjan~JP Jeckmans}, \bibinfo{person}{Michael
  Beye}, \bibinfo{person}{Zekeriya Erkin}, \bibinfo{person}{Pieter Hartel},
  \bibinfo{person}{Reginald~L Lagendijk}, {and} \bibinfo{person}{Qiang Tang}.}
  \bibinfo{year}{2013}\natexlab{}.
\newblock \showarticletitle{Privacy in recommender systems}.
\newblock In \bibinfo{booktitle}{\emph{Social media retrieval}}.
  \bibinfo{publisher}{Springer}, \bibinfo{pages}{263--281}.
\newblock


\bibitem[\protect\citeauthoryear{Jim, Giles, and Horne}{Jim
  et~al\mbox{.}}{1996}]%
        {jim1996analysis}
\bibfield{author}{\bibinfo{person}{Kam-Chuen Jim}, \bibinfo{person}{C~Lee
  Giles}, {and} \bibinfo{person}{Bill~G Horne}.}
  \bibinfo{year}{1996}\natexlab{}.
\newblock \showarticletitle{An analysis of noise in recurrent neural networks:
  convergence and generalization}.
\newblock \bibinfo{journal}{\emph{IEEE {T}ransactions on {N}eural {N}etworks}}
  \bibinfo{volume}{7}, \bibinfo{number}{6} (\bibinfo{year}{1996}),
  \bibinfo{pages}{1424--1438}.
\newblock


\bibitem[\protect\citeauthoryear{Kang and McAuley}{Kang and McAuley}{2018}]%
        {kang2018self}
\bibfield{author}{\bibinfo{person}{Wang-Cheng Kang} {and}
  \bibinfo{person}{Julian McAuley}.} \bibinfo{year}{2018}\natexlab{}.
\newblock \showarticletitle{Self-attentive sequential recommendation}. In
  \bibinfo{booktitle}{\emph{I{C}{D}{M}}}. IEEE, \bibinfo{pages}{197--206}.
\newblock


\bibitem[\protect\citeauthoryear{Kasiviswanathan, Lee, Nissim, Raskhodnikova,
  and Smith}{Kasiviswanathan et~al\mbox{.}}{2011}]%
        {kasiviswanathan2011can}
\bibfield{author}{\bibinfo{person}{Shiva~Prasad Kasiviswanathan},
  \bibinfo{person}{Homin~K Lee}, \bibinfo{person}{Kobbi Nissim},
  \bibinfo{person}{Sofya Raskhodnikova}, {and} \bibinfo{person}{Adam Smith}.}
  \bibinfo{year}{2011}\natexlab{}.
\newblock \showarticletitle{What can we learn privately?}
\newblock \bibinfo{journal}{\emph{S{I}{A}{M} {J}ournal on {C}omputing}}
  \bibinfo{volume}{40}, \bibinfo{number}{3} (\bibinfo{year}{2011}),
  \bibinfo{pages}{793--826}.
\newblock


\bibitem[\protect\citeauthoryear{Kipf and Welling}{Kipf and Welling}{2017}]%
        {DBLP:conf/iclr/KipfW17}
\bibfield{author}{\bibinfo{person}{Thomas~N. Kipf} {and} \bibinfo{person}{Max
  Welling}.} \bibinfo{year}{2017}\natexlab{}.
\newblock \showarticletitle{Semi-Supervised Classification with Graph
  Convolutional Networks}. In \bibinfo{booktitle}{\emph{I{C}{L}{R}}}.
\newblock


\bibitem[\protect\citeauthoryear{Kolluri, Baluta, Hooi, and Saxena}{Kolluri
  et~al\mbox{.}}{2022}]%
        {kolluri2022lpgnet}
\bibfield{author}{\bibinfo{person}{Aashish Kolluri}, \bibinfo{person}{Teodora
  Baluta}, \bibinfo{person}{Bryan Hooi}, {and} \bibinfo{person}{Prateek
  Saxena}.} \bibinfo{year}{2022}\natexlab{}.
\newblock \showarticletitle{LPGNet: Link Private Graph Networks for Node
  Classification}.
\newblock \bibinfo{journal}{\emph{arXiv preprint arXiv:2205.03105}}
  (\bibinfo{year}{2022}).
\newblock


\bibitem[\protect\citeauthoryear{Koren, Bell, and Volinsky}{Koren
  et~al\mbox{.}}{2009}]%
        {koren2009matrix}
\bibfield{author}{\bibinfo{person}{Yehuda Koren}, \bibinfo{person}{Robert
  Bell}, {and} \bibinfo{person}{Chris Volinsky}.}
  \bibinfo{year}{2009}\natexlab{}.
\newblock \showarticletitle{Matrix factorization techniques for recommender
  systems}.
\newblock \bibinfo{journal}{\emph{Computer}} \bibinfo{volume}{42},
  \bibinfo{number}{8} (\bibinfo{year}{2009}), \bibinfo{pages}{30--37}.
\newblock


\bibitem[\protect\citeauthoryear{Kosinski, Stillwell, and Graepel}{Kosinski
  et~al\mbox{.}}{2013}]%
        {kosinski2013private}
\bibfield{author}{\bibinfo{person}{Michal Kosinski}, \bibinfo{person}{David
  Stillwell}, {and} \bibinfo{person}{Thore Graepel}.}
  \bibinfo{year}{2013}\natexlab{}.
\newblock \showarticletitle{Private traits and attributes are predictable from
  digital records of human behavior}.
\newblock \bibinfo{journal}{\emph{Proceedings of the {N}ational {A}cademy of
  {S}ciences}} \bibinfo{volume}{110}, \bibinfo{number}{15}
  (\bibinfo{year}{2013}), \bibinfo{pages}{5802--5805}.
\newblock


\bibitem[\protect\citeauthoryear{Li, Lin, Xiahou, Lin, Wu, and Liu}{Li
  et~al\mbox{.}}{2022}]%
        {li2022federated}
\bibfield{author}{\bibinfo{person}{Li Li}, \bibinfo{person}{Fan Lin},
  \bibinfo{person}{Jianbing Xiahou}, \bibinfo{person}{Yuanguo Lin},
  \bibinfo{person}{Pengcheng Wu}, {and} \bibinfo{person}{Yong Liu}.}
  \bibinfo{year}{2022}\natexlab{}.
\newblock \showarticletitle{Federated low-rank tensor projections for
  sequential recommendation}.
\newblock \bibinfo{journal}{\emph{Knowledge-Based Systems}}
  \bibinfo{volume}{255} (\bibinfo{year}{2022}), \bibinfo{pages}{109483}.
\newblock


\bibitem[\protect\citeauthoryear{Li, Xiahou, Lin, and Su}{Li
  et~al\mbox{.}}{2023}]%
        {li2023distvae}
\bibfield{author}{\bibinfo{person}{Li Li}, \bibinfo{person}{Jianbing Xiahou},
  \bibinfo{person}{Fan Lin}, {and} \bibinfo{person}{Songzhi Su}.}
  \bibinfo{year}{2023}\natexlab{}.
\newblock \showarticletitle{DistVAE: Distributed Variational Autoencoder for
  sequential recommendation}.
\newblock \bibinfo{journal}{\emph{Knowledge-Based Systems}}
  \bibinfo{volume}{264} (\bibinfo{year}{2023}), \bibinfo{pages}{110313}.
\newblock


\bibitem[\protect\citeauthoryear{Li and Kobsa}{Li and Kobsa}{2020}]%
        {li2020context}
\bibfield{author}{\bibinfo{person}{Yao Li} {and} \bibinfo{person}{Alfred
  Kobsa}.} \bibinfo{year}{2020}\natexlab{}.
\newblock \showarticletitle{Context and privacy concerns in friend request
  decisions}.
\newblock \bibinfo{journal}{\emph{Journal of the {A}ssociation for
  {I}nformation {S}cience and {T}echnology}} \bibinfo{volume}{71},
  \bibinfo{number}{6} (\bibinfo{year}{2020}), \bibinfo{pages}{632--643}.
\newblock


\bibitem[\protect\citeauthoryear{Li, Tarlow, Brockschmidt, and Zemel}{Li
  et~al\mbox{.}}{2016}]%
        {DBLP:journals/corr/LiTBZ15}
\bibfield{author}{\bibinfo{person}{Yujia Li}, \bibinfo{person}{Daniel Tarlow},
  \bibinfo{person}{Marc Brockschmidt}, {and} \bibinfo{person}{Richard~S.
  Zemel}.} \bibinfo{year}{2016}\natexlab{}.
\newblock \showarticletitle{Gated Graph Sequence Neural Networks}. In
  \bibinfo{booktitle}{\emph{I{C}{L}{R}}}.
\newblock


\bibitem[\protect\citeauthoryear{Liu, Wang, and Smola}{Liu
  et~al\mbox{.}}{2015}]%
        {liu2015fast}
\bibfield{author}{\bibinfo{person}{Ziqi Liu}, \bibinfo{person}{Yu-Xiang Wang},
  {and} \bibinfo{person}{Alexander Smola}.} \bibinfo{year}{2015}\natexlab{}.
\newblock \showarticletitle{Fast differentially private matrix factorization}.
  In \bibinfo{booktitle}{\emph{Rec{S}ys}}. \bibinfo{pages}{171--178}.
\newblock


\bibitem[\protect\citeauthoryear{Ma, Ma, Zhang, Sun, Liu, and Coates}{Ma
  et~al\mbox{.}}{2020}]%
        {ma2020memory}
\bibfield{author}{\bibinfo{person}{Chen Ma}, \bibinfo{person}{Liheng Ma},
  \bibinfo{person}{Yingxue Zhang}, \bibinfo{person}{Jianing Sun},
  \bibinfo{person}{Xue Liu}, {and} \bibinfo{person}{Mark Coates}.}
  \bibinfo{year}{2020}\natexlab{}.
\newblock \showarticletitle{Memory augmented graph neural networks for
  sequential recommendation}. In \bibinfo{booktitle}{\emph{A{A}{A}{I}}},
  Vol.~\bibinfo{volume}{34}. \bibinfo{pages}{5045--5052}.
\newblock


\bibitem[\protect\citeauthoryear{McSherry and Mironov}{McSherry and
  Mironov}{2009}]%
        {mcsherry2009differentially}
\bibfield{author}{\bibinfo{person}{Frank McSherry} {and} \bibinfo{person}{Ilya
  Mironov}.} \bibinfo{year}{2009}\natexlab{}.
\newblock \showarticletitle{Differentially private recommender systems:
  Building privacy into the netflix prize contenders}. In
  \bibinfo{booktitle}{\emph{K{D}{D}}}. \bibinfo{pages}{627--636}.
\newblock


\bibitem[\protect\citeauthoryear{Menkov, Ginsparg, and Kantor}{Menkov
  et~al\mbox{.}}{2020}]%
        {menkov2020recommendations}
\bibfield{author}{\bibinfo{person}{Vladimir Menkov}, \bibinfo{person}{Paul
  Ginsparg}, {and} \bibinfo{person}{Paul~B Kantor}.}
  \bibinfo{year}{2020}\natexlab{}.
\newblock \showarticletitle{Recommendations and privacy in the arXiv system: A
  simulation experiment using historical data}.
\newblock \bibinfo{journal}{\emph{Journal of the Association for Information
  Science and Technology}} \bibinfo{volume}{71}, \bibinfo{number}{3}
  (\bibinfo{year}{2020}), \bibinfo{pages}{300--313}.
\newblock


\bibitem[\protect\citeauthoryear{Mironov}{Mironov}{2017}]%
        {mironov2017renyi}
\bibfield{author}{\bibinfo{person}{Ilya Mironov}.}
  \bibinfo{year}{2017}\natexlab{}.
\newblock \showarticletitle{R{\'e}nyi differential privacy}. In
  \bibinfo{booktitle}{\emph{C{S}{F}}}. \bibinfo{pages}{263--275}.
\newblock


\bibitem[\protect\citeauthoryear{Nikolaenko, Ioannidis, Weinsberg, Joye, Taft,
  and Boneh}{Nikolaenko et~al\mbox{.}}{2013}]%
        {nikolaenko2013privacy}
\bibfield{author}{\bibinfo{person}{Valeria Nikolaenko},
  \bibinfo{person}{Stratis Ioannidis}, \bibinfo{person}{Udi Weinsberg},
  \bibinfo{person}{Marc Joye}, \bibinfo{person}{Nina Taft}, {and}
  \bibinfo{person}{Dan Boneh}.} \bibinfo{year}{2013}\natexlab{}.
\newblock \showarticletitle{Privacy-preserving matrix factorization}. In
  \bibinfo{booktitle}{\emph{Proceedings of the 2013 ACM SIGSAC conference on
  Computer \& communications security}}. \bibinfo{pages}{801--812}.
\newblock


\bibitem[\protect\citeauthoryear{Quadrana, Karatzoglou, Hidasi, and
  Cremonesi}{Quadrana et~al\mbox{.}}{2017}]%
        {quadrana2017personalizing}
\bibfield{author}{\bibinfo{person}{Massimo Quadrana},
  \bibinfo{person}{Alexandros Karatzoglou}, \bibinfo{person}{Bal{\'a}zs
  Hidasi}, {and} \bibinfo{person}{Paolo Cremonesi}.}
  \bibinfo{year}{2017}\natexlab{}.
\newblock \showarticletitle{Personalizing session-based recommendations with
  hierarchical recurrent neural networks}. In
  \bibinfo{booktitle}{\emph{Rec{S}ys}}. \bibinfo{pages}{130--137}.
\newblock


\bibitem[\protect\citeauthoryear{Ren, Jiang, Peng, Lyu, Liu, Chen, Wu, Bai, and
  Yu}{Ren et~al\mbox{.}}{2022}]%
        {ren2022cross}
\bibfield{author}{\bibinfo{person}{Jiaqian Ren}, \bibinfo{person}{Lei Jiang},
  \bibinfo{person}{Hao Peng}, \bibinfo{person}{Lingjuan Lyu},
  \bibinfo{person}{Zhiwei Liu}, \bibinfo{person}{Chaochao Chen},
  \bibinfo{person}{Jia Wu}, \bibinfo{person}{Xu Bai}, {and}
  \bibinfo{person}{Philip~S Yu}.} \bibinfo{year}{2022}\natexlab{}.
\newblock \showarticletitle{Cross-Network Social User Embedding with Hybrid
  Differential Privacy Guarantees}. In \bibinfo{booktitle}{\emph{Proceedings of
  the 31st ACM International Conference on Information \& Knowledge
  Management}}. \bibinfo{pages}{1685--1695}.
\newblock


\bibitem[\protect\citeauthoryear{Rendle, Freudenthaler, Gantner, and
  Schmidt{-}Thieme}{Rendle et~al\mbox{.}}{2009}]%
        {DBLP:conf/uai/RendleFGS09}
\bibfield{author}{\bibinfo{person}{Steffen Rendle}, \bibinfo{person}{Christoph
  Freudenthaler}, \bibinfo{person}{Zeno Gantner}, {and} \bibinfo{person}{Lars
  Schmidt{-}Thieme}.} \bibinfo{year}{2009}\natexlab{}.
\newblock \showarticletitle{{BPR:} Bayesian Personalized Ranking from Implicit
  Feedback}. In \bibinfo{booktitle}{\emph{U{A}{I}}}. \bibinfo{pages}{452--461}.
\newblock


\bibitem[\protect\citeauthoryear{Rendle, Freudenthaler, and
  Schmidt{-}Thieme}{Rendle et~al\mbox{.}}{2010}]%
        {DBLP:conf/www/RendleFS10}
\bibfield{author}{\bibinfo{person}{Steffen Rendle}, \bibinfo{person}{Christoph
  Freudenthaler}, {and} \bibinfo{person}{Lars Schmidt{-}Thieme}.}
  \bibinfo{year}{2010}\natexlab{}.
\newblock \showarticletitle{Factorizing personalized Markov chains for
  next-basket recommendation}. In \bibinfo{booktitle}{\emph{W{W}{W}}}.
  \bibinfo{pages}{811--820}.
\newblock


\bibitem[\protect\citeauthoryear{Sajadmanesh, Shamsabadi, Bellet, and
  Gatica-Perez}{Sajadmanesh et~al\mbox{.}}{2022}]%
        {sajadmanesh2022gap}
\bibfield{author}{\bibinfo{person}{Sina Sajadmanesh},
  \bibinfo{person}{Ali~Shahin Shamsabadi}, \bibinfo{person}{Aur{\'e}lien
  Bellet}, {and} \bibinfo{person}{Daniel Gatica-Perez}.}
  \bibinfo{year}{2022}\natexlab{}.
\newblock \showarticletitle{GAP: Differentially Private Graph Neural Networks
  with Aggregation Perturbation}.
\newblock \bibinfo{journal}{\emph{arXiv preprint arXiv:2203.00949}}
  (\bibinfo{year}{2022}).
\newblock


\bibitem[\protect\citeauthoryear{Shani, Heckerman, Brafman, and
  Boutilier}{Shani et~al\mbox{.}}{2005}]%
        {shani2005mdp}
\bibfield{author}{\bibinfo{person}{Guy Shani}, \bibinfo{person}{David
  Heckerman}, \bibinfo{person}{Ronen~I Brafman}, {and} \bibinfo{person}{Craig
  Boutilier}.} \bibinfo{year}{2005}\natexlab{}.
\newblock \showarticletitle{An MDP-based recommender system.}
\newblock \bibinfo{journal}{\emph{Journal of {M}achine {L}earning {R}esearch}}
  \bibinfo{volume}{6}, \bibinfo{number}{9} (\bibinfo{year}{2005}).
\newblock


\bibitem[\protect\citeauthoryear{Shin, Kim, Shin, and Xiao}{Shin
  et~al\mbox{.}}{2018}]%
        {shin2018privacy}
\bibfield{author}{\bibinfo{person}{Hyejin Shin}, \bibinfo{person}{Sungwook
  Kim}, \bibinfo{person}{Junbum Shin}, {and} \bibinfo{person}{Xiaokui Xiao}.}
  \bibinfo{year}{2018}\natexlab{}.
\newblock \showarticletitle{Privacy enhanced matrix factorization for
  recommendation with local differential privacy}.
\newblock \bibinfo{journal}{\emph{T{K}{D}{E}}} \bibinfo{volume}{30},
  \bibinfo{number}{9} (\bibinfo{year}{2018}), \bibinfo{pages}{1770--1782}.
\newblock


\bibitem[\protect\citeauthoryear{Sun, Liu, Wu, Pei, Lin, Ou, and Jiang}{Sun
  et~al\mbox{.}}{2019}]%
        {sun2019bert4rec}
\bibfield{author}{\bibinfo{person}{Fei Sun}, \bibinfo{person}{Jun Liu},
  \bibinfo{person}{Jian Wu}, \bibinfo{person}{Changhua Pei},
  \bibinfo{person}{Xiao Lin}, \bibinfo{person}{Wenwu Ou}, {and}
  \bibinfo{person}{Peng Jiang}.} \bibinfo{year}{2019}\natexlab{}.
\newblock \showarticletitle{BERT4Rec: Sequential recommendation with
  bidirectional encoder representations from transformer}. In
  \bibinfo{booktitle}{\emph{Proceedings of the 28th ACM international
  {C}onference on {I}nformation and {K}nowledge {M}anagement}}.
  \bibinfo{pages}{1441--1450}.
\newblock


\bibitem[\protect\citeauthoryear{Vaswani, Shazeer, Parmar, Uszkoreit, Jones,
  Gomez, Kaiser, and Polosukhin}{Vaswani et~al\mbox{.}}{2017}]%
        {vaswani2017attention}
\bibfield{author}{\bibinfo{person}{Ashish Vaswani}, \bibinfo{person}{Noam
  Shazeer}, \bibinfo{person}{Niki Parmar}, \bibinfo{person}{Jakob Uszkoreit},
  \bibinfo{person}{Llion Jones}, \bibinfo{person}{Aidan~N Gomez},
  \bibinfo{person}{{\L}ukasz Kaiser}, {and} \bibinfo{person}{Illia
  Polosukhin}.} \bibinfo{year}{2017}\natexlab{}.
\newblock \showarticletitle{Attention is all you need}.
\newblock \bibinfo{journal}{\emph{Advances in {N}eural {I}nformation
  {P}rocessing {S}ystems}}  \bibinfo{volume}{30} (\bibinfo{year}{2017}).
\newblock


\bibitem[\protect\citeauthoryear{Veli{\v{c}}kovi{\'c}, Cucurull, Casanova,
  Romero, Li{\`o}, and Bengio}{Veli{\v{c}}kovi{\'c} et~al\mbox{.}}{2018}]%
        {velivckovic2018graph}
\bibfield{author}{\bibinfo{person}{Petar Veli{\v{c}}kovi{\'c}},
  \bibinfo{person}{Guillem Cucurull}, \bibinfo{person}{Arantxa Casanova},
  \bibinfo{person}{Adriana Romero}, \bibinfo{person}{Pietro Li{\`o}}, {and}
  \bibinfo{person}{Yoshua Bengio}.} \bibinfo{year}{2018}\natexlab{}.
\newblock \showarticletitle{Graph Attention Networks}. In
  \bibinfo{booktitle}{\emph{I{C}{L}{R}}}.
\newblock


\bibitem[\protect\citeauthoryear{Voigt and Von~dem Bussche}{Voigt and Von~dem
  Bussche}{2017}]%
        {voigt2017eu}
\bibfield{author}{\bibinfo{person}{Paul Voigt} {and} \bibinfo{person}{Axel
  Von~dem Bussche}.} \bibinfo{year}{2017}\natexlab{}.
\newblock \showarticletitle{The {EU} general data protection regulation
  ({GDPR})}.
\newblock \bibinfo{journal}{\emph{A {P}ractical {G}uide, 1st Ed., Cham:
  {S}pringer {I}nternational {P}ublishing}} \bibinfo{volume}{10},
  \bibinfo{number}{3152676} (\bibinfo{year}{2017}), \bibinfo{pages}{10--5555}.
\newblock


\bibitem[\protect\citeauthoryear{Wang, Guo, Li, Chen, and Li}{Wang
  et~al\mbox{.}}{2021}]%
        {wang2021privacy}
\bibfield{author}{\bibinfo{person}{Binghui Wang}, \bibinfo{person}{Jiayi Guo},
  \bibinfo{person}{Ang Li}, \bibinfo{person}{Yiran Chen}, {and}
  \bibinfo{person}{Hai Li}.} \bibinfo{year}{2021}\natexlab{}.
\newblock \showarticletitle{Privacy-preserving representation learning on
  graphs: A mutual information perspective}. In
  \bibinfo{booktitle}{\emph{K{D}{D}}}. \bibinfo{pages}{1667--1676}.
\newblock


\bibitem[\protect\citeauthoryear{Wang, Xiao, Yang, Zhao, Hui, Shin, Shin, and
  Yu}{Wang et~al\mbox{.}}{2019b}]%
        {wang2019collecting}
\bibfield{author}{\bibinfo{person}{Ning Wang}, \bibinfo{person}{Xiaokui Xiao},
  \bibinfo{person}{Yin Yang}, \bibinfo{person}{Jun Zhao},
  \bibinfo{person}{Siu~Cheung Hui}, \bibinfo{person}{Hyejin Shin},
  \bibinfo{person}{Junbum Shin}, {and} \bibinfo{person}{Ge Yu}.}
  \bibinfo{year}{2019}\natexlab{b}.
\newblock \showarticletitle{Collecting and analyzing multidimensional data with
  local differential privacy}. In \bibinfo{booktitle}{\emph{I{C}{D}{E}}}. IEEE,
  \bibinfo{pages}{638--649}.
\newblock


\bibitem[\protect\citeauthoryear{Wang, Hu, Wang, Cao, Sheng, and Orgun}{Wang
  et~al\mbox{.}}{2019a}]%
        {wang2019sequential}
\bibfield{author}{\bibinfo{person}{Shoujin Wang}, \bibinfo{person}{Liang Hu},
  \bibinfo{person}{Yan Wang}, \bibinfo{person}{Longbing Cao},
  \bibinfo{person}{Quan~Z Sheng}, {and} \bibinfo{person}{Mehmet Orgun}.}
  \bibinfo{year}{2019}\natexlab{a}.
\newblock \showarticletitle{Sequential recommender systems: challenges,
  progress and prospects}. In \bibinfo{booktitle}{\emph{I{J}{C}{A}{I}}}.
  \bibinfo{pages}{6332--6338}.
\newblock


\bibitem[\protect\citeauthoryear{Wang, Zhang, Hu, Zhang, Wang, and
  Aggarwal}{Wang et~al\mbox{.}}{2022}]%
        {DBLP:conf/sigir/WangZ0ZW022}
\bibfield{author}{\bibinfo{person}{Shoujin Wang}, \bibinfo{person}{Qi Zhang},
  \bibinfo{person}{Liang Hu}, \bibinfo{person}{Xiuzhen Zhang},
  \bibinfo{person}{Yan Wang}, {and} \bibinfo{person}{Charu Aggarwal}.}
  \bibinfo{year}{2022}\natexlab{}.
\newblock \showarticletitle{Sequential/Session-based Recommendations:
  Challenges, Approaches, Applications and Opportunities}. In
  \bibinfo{booktitle}{\emph{{SIGIR}}}. \bibinfo{pages}{3425--3428}.
\newblock


\bibitem[\protect\citeauthoryear{Wang, Blocki, Li, and Jha}{Wang
  et~al\mbox{.}}{2017}]%
        {wang2017locally}
\bibfield{author}{\bibinfo{person}{Tianhao Wang}, \bibinfo{person}{Jeremiah
  Blocki}, \bibinfo{person}{Ninghui Li}, {and} \bibinfo{person}{Somesh Jha}.}
  \bibinfo{year}{2017}\natexlab{}.
\newblock \showarticletitle{Locally differentially private protocols for
  frequency estimation}. In \bibinfo{booktitle}{\emph{U{S}{E}{N}{I}{X}
  {S}ecurity}}. \bibinfo{pages}{729--745}.
\newblock


\bibitem[\protect\citeauthoryear{Wei, Meng, Li, Zhou, Qi, and Xu}{Wei
  et~al\mbox{.}}{2023}]%
        {wei2023edge}
\bibfield{author}{\bibinfo{person}{Shanming Wei}, \bibinfo{person}{Shunmei
  Meng}, \bibinfo{person}{Qianmu Li}, \bibinfo{person}{Xiaokang Zhou},
  \bibinfo{person}{Lianyong Qi}, {and} \bibinfo{person}{Xiaolong Xu}.}
  \bibinfo{year}{2023}\natexlab{}.
\newblock \showarticletitle{Edge-enabled federated sequential recommendation
  with knowledge-aware Transformer}.
\newblock \bibinfo{journal}{\emph{Future Generation Computer Systems}}
  \bibinfo{volume}{148} (\bibinfo{year}{2023}), \bibinfo{pages}{610--622}.
\newblock


\bibitem[\protect\citeauthoryear{Weinsberg, Bhagat, Ioannidis, and
  Taft}{Weinsberg et~al\mbox{.}}{2012}]%
        {weinsberg2012blurme}
\bibfield{author}{\bibinfo{person}{Udi Weinsberg}, \bibinfo{person}{Smriti
  Bhagat}, \bibinfo{person}{Stratis Ioannidis}, {and} \bibinfo{person}{Nina
  Taft}.} \bibinfo{year}{2012}\natexlab{}.
\newblock \showarticletitle{BlurMe: Inferring and obfuscating user gender based
  on ratings}. In \bibinfo{booktitle}{\emph{Rec{S}ys}}.
  \bibinfo{pages}{195--202}.
\newblock


\bibitem[\protect\citeauthoryear{Wu, Long, Zhang, and Li}{Wu
  et~al\mbox{.}}{2022}]%
        {wu2021linkteller}
\bibfield{author}{\bibinfo{person}{Fan Wu}, \bibinfo{person}{Yunhui Long},
  \bibinfo{person}{Ce Zhang}, {and} \bibinfo{person}{Bo Li}.}
  \bibinfo{year}{2022}\natexlab{}.
\newblock \showarticletitle{Linkteller: Recovering private edges from graph
  neural networks via influence analysis}.
\newblock \bibinfo{journal}{\emph{S{P}}} (\bibinfo{year}{2022}).
\newblock


\bibitem[\protect\citeauthoryear{Wu, Tang, Zhu, Wang, Xie, and Tan}{Wu
  et~al\mbox{.}}{2019}]%
        {wu2019session}
\bibfield{author}{\bibinfo{person}{Shu Wu}, \bibinfo{person}{Yuyuan Tang},
  \bibinfo{person}{Yanqiao Zhu}, \bibinfo{person}{Liang Wang},
  \bibinfo{person}{Xing Xie}, {and} \bibinfo{person}{Tieniu Tan}.}
  \bibinfo{year}{2019}\natexlab{}.
\newblock \showarticletitle{Session-based recommendation with graph neural
  networks}. In \bibinfo{booktitle}{\emph{A{A}{A}{I}}},
  Vol.~\bibinfo{volume}{33}. \bibinfo{pages}{346--353}.
\newblock


\bibitem[\protect\citeauthoryear{Xu, Zhao, Liu, Xu, S.~Sheng, Cui, Zhou, and
  Xiong}{Xu et~al\mbox{.}}{2019}]%
        {xu2019recurrent}
\bibfield{author}{\bibinfo{person}{Chengfeng Xu}, \bibinfo{person}{Pengpeng
  Zhao}, \bibinfo{person}{Yanchi Liu}, \bibinfo{person}{Jiajie Xu},
  \bibinfo{person}{Victor S~Sheng S.~Sheng}, \bibinfo{person}{Zhiming Cui},
  \bibinfo{person}{Xiaofang Zhou}, {and} \bibinfo{person}{Hui Xiong}.}
  \bibinfo{year}{2019}\natexlab{}.
\newblock \showarticletitle{Recurrent convolutional neural network for
  sequential recommendation}. In \bibinfo{booktitle}{\emph{W{W}{W}}}.
  \bibinfo{pages}{3398--3404}.
\newblock


\bibitem[\protect\citeauthoryear{Ying, Zhuang, Zhang, Liu, Xu, Xie, Xiong, and
  Wu}{Ying et~al\mbox{.}}{2018}]%
        {ying2018sequential}
\bibfield{author}{\bibinfo{person}{Haochao Ying}, \bibinfo{person}{Fuzhen
  Zhuang}, \bibinfo{person}{Fuzheng Zhang}, \bibinfo{person}{Yanchi Liu},
  \bibinfo{person}{Guandong Xu}, \bibinfo{person}{Xing Xie},
  \bibinfo{person}{Hui Xiong}, {and} \bibinfo{person}{Jian Wu}.}
  \bibinfo{year}{2018}\natexlab{}.
\newblock \showarticletitle{Sequential recommender system based on hierarchical
  attention network}. In \bibinfo{booktitle}{\emph{I{J}{C}{A}{I}}}.
\newblock


\bibitem[\protect\citeauthoryear{Zhang, Wu, Gao, Jiang, Xu, and Wang}{Zhang
  et~al\mbox{.}}{2020}]%
        {zhang2020personalized}
\bibfield{author}{\bibinfo{person}{Mengqi Zhang}, \bibinfo{person}{Shu Wu},
  \bibinfo{person}{Meng Gao}, \bibinfo{person}{Xin Jiang}, \bibinfo{person}{Ke
  Xu}, {and} \bibinfo{person}{Liang Wang}.} \bibinfo{year}{2020}\natexlab{}.
\newblock \showarticletitle{Personalized graph neural networks with attention
  mechanism for session-aware recommendation}.
\newblock \bibinfo{journal}{\emph{T{K}{D}{E}}} (\bibinfo{year}{2020}).
\newblock


\bibitem[\protect\citeauthoryear{Zhang, Yin, Chen, Huang, Cui, and Zhang}{Zhang
  et~al\mbox{.}}{2021b}]%
        {zhang2021graph}
\bibfield{author}{\bibinfo{person}{Shijie Zhang}, \bibinfo{person}{Hongzhi
  Yin}, \bibinfo{person}{Tong Chen}, \bibinfo{person}{Zi Huang},
  \bibinfo{person}{Lizhen Cui}, {and} \bibinfo{person}{Xiangliang Zhang}.}
  \bibinfo{year}{2021}\natexlab{b}.
\newblock \showarticletitle{Graph embedding for recommendation against
  attribute inference attacks}. In \bibinfo{booktitle}{\emph{W{W}{W}}}.
  \bibinfo{pages}{3002--3014}.
\newblock


\bibitem[\protect\citeauthoryear{Zhang, Zhao, Liu, Sheng, Xu, Wang, Liu, and
  Zhou}{Zhang et~al\mbox{.}}{2019}]%
        {zhang2019feature}
\bibfield{author}{\bibinfo{person}{Tingting Zhang}, \bibinfo{person}{Pengpeng
  Zhao}, \bibinfo{person}{Yanchi Liu}, \bibinfo{person}{Victor~S Sheng},
  \bibinfo{person}{Jiajie Xu}, \bibinfo{person}{Deqing Wang},
  \bibinfo{person}{Guanfeng Liu}, {and} \bibinfo{person}{Xiaofang Zhou}.}
  \bibinfo{year}{2019}\natexlab{}.
\newblock \showarticletitle{Feature-level Deeper Self-Attention Network for
  Sequential Recommendation.}. In \bibinfo{booktitle}{\emph{I{J}{C}{A}{I}}}.
  \bibinfo{pages}{4320--4326}.
\newblock


\bibitem[\protect\citeauthoryear{Zhang, Liu, Huang, Wang, Lu, Liu, and
  Chen}{Zhang et~al\mbox{.}}{2021a}]%
        {zhang2021graphmi}
\bibfield{author}{\bibinfo{person}{Zaixi Zhang}, \bibinfo{person}{Qi Liu},
  \bibinfo{person}{Zhenya Huang}, \bibinfo{person}{Hao Wang},
  \bibinfo{person}{Chengqiang Lu}, \bibinfo{person}{Chuanren Liu}, {and}
  \bibinfo{person}{Enhong Chen}.} \bibinfo{year}{2021}\natexlab{a}.
\newblock \showarticletitle{Graphmi: Extracting private graph data from graph
  neural networks}.
\newblock \bibinfo{journal}{\emph{I{J}{C}{A}{I}}} (\bibinfo{year}{2021}).
\newblock


\bibitem[\protect\citeauthoryear{Zhu, Li, Ren, Zhou, and Xiong}{Zhu
  et~al\mbox{.}}{2013}]%
        {zhu2013differential}
\bibfield{author}{\bibinfo{person}{Tianqing Zhu}, \bibinfo{person}{Gang Li},
  \bibinfo{person}{Yongli Ren}, \bibinfo{person}{Wanlei Zhou}, {and}
  \bibinfo{person}{Ping Xiong}.} \bibinfo{year}{2013}\natexlab{}.
\newblock \showarticletitle{Differential privacy for neighborhood-based
  collaborative filtering}. In \bibinfo{booktitle}{\emph{Proceedings of the
  2013 IEEE/ACM {I}nternational {C}onference on {A}dvances in {S}ocial
  {N}etworks {A}nalysis and {M}ining}}. \bibinfo{pages}{752--759}.
\newblock


\end{thebibliography}

%\appendix

\end{document}